\documentclass[smallextended]{svjour3}
\smartqed
\usepackage{balance}
\usepackage{microtype}
\usepackage{enumitem}
\usepackage{xspace}
\usepackage{xcolor}
\usepackage{hyperref}
\usepackage{boxedminipage}
\usepackage{amsmath,amsfonts,latexsym}
\usepackage{graphicx}
\usepackage[sort]{cite}
\usepackage{comment}

\renewcommand{\paragraph}[1]{\vspace{6pt}\noindent\textbf{#1}}

\newcommand{\poly}{\ensuremath{{\sf poly}}\xspace}
\newcommand{\ppt}{\ensuremath{\text{p.p.t.}}\xspace}
\newcommand{\msg}{\ensuremath{{{\sf m}}}\xspace}

\newcommand{\comm}{\ensuremath{{{\sf com}}}\xspace}

\newcommand{\mcal}[1]{\ensuremath{\mathcal {#1}}}

\newcommand{\cc}{\ensuremath{{C}}\xspace}
\newcommand{\mcc}{\ensuremath{C}\xspace}

\newcommand{\nizk}{\ensuremath{{\sf nizk}}\xspace}
\newcommand{\prf}{\ensuremath{{\sf PRF}}\xspace}

\newcommand{\Coin}{\ensuremath{{{\sf Coin}}}\xspace}

\newcommand{\sk}{\mathsf{sk}}
\newcommand{\pk}{\mathsf{pk}}

\renewcommand{\k}{\ensuremath{\secparam}}
\newcommand{\secparam}{\ensuremath{{\kappa}}}
\newcommand{\numrounds}{\ensuremath{{\lambda}}}
\newcommand{\comsize}{\ensuremath{{\lambda}}}
\newcommand{\compsec}{\ensuremath{{\chi}}}

\newcommand{\algA}{{\color{black}\ensuremath{\mcal{A}}}\xspace}

\newcommand{\algC}{{\color{black}\ensuremath{\mcal{C}}}\xspace}

\newcommand{\algR}{{\color{black}\ensuremath{\mcal{R}}}\xspace}
\newcommand{\algZ}{{\ensuremath{{\mcal{Z}}}}\xspace}
\newcommand{\AZ}{{\ensuremath{({\mcal{A}}, \mcal{Z})}}\xspace}

\newcommand{\Fsort}{\ensuremath{\mcal{F}_{\rm mine}}\xspace}
\newcommand{\Fmine}{{\color{black}\ensuremath{\mcal{F}_{\rm mine}}}\xspace}
\newcommand{\mine}{{\ensuremath{\tt mine}}\xspace}
\newcommand{\ver}{{\ensuremath{\tt verify}}\xspace}

\newcommand{\negl}{{\sf negl}}

\newcommand{\lang}{{\ensuremath{{\mathcal{L}}}}\xspace}
\newcommand{\keygen}{{\ensuremath{{\sf Gen}}}\xspace}

\newcommand{\view}{\textsf{view}}

\newcommand{\N}{\mathbb{N}}

\newcommand{\protideal}{\ensuremath{\Pi_{\rm ideal}}\xspace}

\newcommand{\crs}{{\ensuremath{{\sf crs}}}\xspace}
\newcommand{\wtcrs}{{\ensuremath{\overline{\sf crs}}}\xspace}
\newcommand{\trap}{{\ensuremath{\tau}}\xspace}

\newcommand{\epk}{{\ensuremath{{\sf epk}}}\xspace}

\newcommand{\stmt}{{\ensuremath{{\sf stmt}}}\xspace}

\newcommand{\exec}{\textsf{EXEC}}

\newcommand{\ignore}[1]{}
\newcommand{\ignorepsync}[1]{}

\newcommand{\elaine}[1]{{\footnotesize\color{magenta}[Elaine:
#1]}}

\newcommand{\hubert}[1]{{\footnotesize\color{red}[Hubert: #1]}}
\newcommand{\rl}[1]{{\footnotesize\color{cyan}[Ling: #1]}}
\newcommand{\kartik}[1]{{\footnotesize\color{blue}[Kartik: #1]}}

\renewcommand{\elaine}[1]{}
\renewcommand{\rl}[1]{}

\newcommand{\hidex}[1]{#1}

\newcommand{\MType}{{\ensuremath{\tt T}}\xspace}
\newcommand{\Status}{{\ensuremath{\tt Status}}\xspace}
\newcommand{\Propose}{{\ensuremath{\tt Propose}}\xspace}

\newcommand{\Vote}{{\ensuremath{\tt Vote}}\xspace}
\newcommand{\Commit}{{\ensuremath{\tt Commit}}\xspace}

\newcommand{\Terminate}{{\ensuremath{\tt Terminate}}\xspace}
\newcommand{\Cert}{{\ensuremath{\mathcal{C}}}\xspace}
\newcommand{\WC}[1]{{\ensuremath{W_{c,#1}}}}
\newcommand{\WH}[1]{{\ensuremath{W_{h,#1}}}}
\newcommand{\nodeS}{node~$S$\xspace}

\begin{document}

\title{Communication Complexity of Byzantine Agreement,
  Revisited\thanks{This work is partially supported by The Federmann Cyber Security Center in conjunction with the Israel National Cyber Directorate.
	T-H. Hubert Chan was partially supported by the Hong Kong RGC under the grant 17200418.
		}}
\author{Ittai Abraham \and
        T-H.\ Hubert Chan \and
        Danny Dolev \and
        Kartik Nayak \and
        Rafael Pass \and
        Ling Ren \and
        Elaine Shi}

%\institute{F. Author \at
%              first address \\
%              Tel.: +123-45-678910\\
%              Fax: +123-45-678910\\
%              \email{fauthor@example.com}           %  \\
\institute{VMware Research \at \email{iabraham@vmware.com} \and
           The University of Hong Kong \at
           \email{hubert@cs.hku.hk} \and
           The Hebrew University of Jerusalem \at
           \email{dolev@cs.huji.ac.il} \and
           Duke University \at \email{kartik@cs.duke.edu} \and
           Cornell Tech \at \email{rafael@cs.cornell.edu} \and
           University of Illinois, Urbana-Champaign \at
           \email{renling@illinois.edu} \and 
           Cornell University \at \email{runting@gmail.com}
}

%\ignore{
%\settopmatter{authorsperrow=4}
%\author{Ittai Abraham}
%\affiliation{VMware Research}
%\author{T-H. Hubert Chan}
%\affiliation{The University of Hong Kong}
%\author{Danny Dolev}
%\affiliation{The Hebrew University of Jerusalem} 
%\author{Kartik Nayak}
%\affiliation{Duke University}
%\author{Rafael Pass}
%\affiliation{Cornell Tech}
%\author{Ling Ren}
%\affiliation{VMware Research}
%\author{Elaine Shi}
%\affiliation{Cornell University}

% Ittai Abraham \and T-H. Hubert Chan \and Danny Dolev \and Kartik Nayak \and Rafael Pass \and Ling Ren \and Elaine Shi

\date{}
\maketitle

\begin{abstract}
As Byzantine Agreement (BA) protocols find application in large-scale decentralized cryptocurrencies, an increasingly important problem is to design BA protocols with improved communication complexity.
A few existing works have shown how to achieve subquadratic BA under an {\it adaptive} adversary.
%These include King and Saia (PODC'10),  Nakamoto consensus, and a few proof-of-stake protocols. 
Intriguingly, they all make a common relaxation about the adaptivity of the attacker, that is, if an honest node sends a message and then gets corrupted in some round, the adversary {\it cannot erase the message that was already sent} --- henceforth we say that such an adversary cannot perform ``after-the-fact removal''. 
By contrast, many (super-)quadratic BA protocols in the literature can tolerate after-the-fact removal.
In this paper, we first prove that disallowing after-the-fact removal is necessary for achieving subquadratic-communication BA. 

Next, we show new subquadratic binary BA constructions
(of course, assuming no after-the-fact removal)
that achieves near-optimal resilience and expected constant rounds under 
standard cryptographic assumptions and a public-key infrastructure (PKI) in both synchronous and partially synchronous settings. 
In comparison, all known subquadratic protocols make additional strong assumptions such as random oracles or the ability of honest nodes to erase secrets from memory,
and even with these strong assumptions, no prior work can achieve the above properties.
Lastly, we show that some setup assumption is necessary for achieving subquadratic multicast-based BA. 
%(i.e., all protocol messages are multicast to other nodes).

%%% Local Variables:
%%% mode: plain-tex
%%% TeX-master: "podc2019"
%%% End:

\end{abstract}

%\sloppy
\section{Introduction}
\label{sec:intro}

\ignore{
Roughly speaking, in (binary) BA, 
every player receives an input bit that is either 0 or 1.
The players' goal is to agree on a bit
such that the following properties are satisfied: 
\begin{itemize}[leftmargin=12pt]
%\item[-] \emph{termination:} all honest players output 
\item[-] \emph{consistency/safety:} all honest players output the same bit, 
\item[-] \emph{validity:} if all honest nodes receive the same input bit $b$, 
then all honest nodes output $b$.\kartik{why not go the standard
  route of listing termination?}
\end{itemize} 
}

Byzantine agreement (BA)~\cite{byzgen}
is a central abstraction in distributed systems.
Typical BA protocols~\cite{dolevstrong,pbft,dworklynch} require all players to
send messages to all other players, and thus, $n$-player BA requires at
least $n^2$ communication complexity. 
Such protocols are thus not well
suited for \emph{large-scale distributed systems} such as decentralized cryptocurrencies~\cite{bitcoin}. 
A fundamental problem is to design BA protocols with improved communication complexity. 

In fact, in a model with {\it static} corruption, this is relatively easy. 
For example, suppose there are at most $f<(\frac 12 -\epsilon)n$ corrupt nodes 
where $\epsilon$ is a positive constant;
further, assume there is a trusted common random string (CRS) that is chosen independently 
of the adversary's (static) corruption choices.  
Then, we can use the CRS to select a $\comsize$-sized committee of players.
Various elegant works have investigated how to weaken or remove the trusted set-up assumptions required for committee election and retain subquadratic communication~\cite{kingsaialeader,hiddengraph}.
Once a committee is selected,
we can run any BA protocol among the committee,
and let the committee members may send their outputs to all ``non-committee''
players who could then output the majority bit. 
This protocol works as long as there is an honest majority on the committee.
Thus, the error probability is bounded by $\exp(-\Omega(\comsize))$ due to a standard Chernoff bound.

Such a committee-based approach, however, fails if we consider an
\emph{adaptive attacker}. Such an attacker can simply observe what nodes are
on the committee, then corrupt them, and thereby control the
entire committee! A natural and long-standing open question is thus
whether subquadratic communication is possible w.r.t. an adaptive attacker:
\begin{itemize}[leftmargin=*]
\item[]
\emph{Does there exist a BA protocol with subquadratic communication
complexity that resists adaptive corruption of players?}
\end{itemize} 

This question has been partially answered in a few prior works.
First, a breakthrough work by King and Saia~\cite{kingsaiabarrier}
presented a BA protocol with communication complexity $O(n^{1.5})$.
More recent works studied practical constructions motivated 
by cryptocurrency applications: notably the celebrated Nakamoto 
consensus~\cite{bitcoin,backbone} can reach agreement in
subquadratic communication assuming idealized proof-of-work.
Subsequently, several so-called ``proof-of-stake'' constructions~\cite{algorand,aggelospos2}
also showed how to realize BA with subquadratic communication. 
All of the above works tolerate adaptive corruptions.

What is both intriguing and unsatisfying is that all these works
happen to make a common relaxing assumption about the adaptivity of the adversary, 
namely,
if adversary adaptively corrupts an honest node $i$ 
who has just sent a message $\msg$ in round $r$, the adversary is unable
to erase the honest message $\msg$ sent in round $r$.
Henceforth we say that such an adversary is incapable of {\it after-the-fact removal}.
In comparison, many natural $\Omega(n^2)$-communication BA protocols~\cite{dolevstrong,katzbft,abraham2018synchronous}
can be proven secure even if the adversary is capable of after-the-fact removal
-- henceforth referred to as a \emph{strongly adaptive} adversary.
That is, 
if an honest node $i$ sends a message $\msg$ in round $r$, 
a strongly adaptive adversary (e.g., who controls the egress routers of many nodes)
can observe $\msg$ and then decide to corrupt $i$ and erase the message $\msg$ that
node~$i$ has just sent in round $r$.
This mismatch in model naturally raises the following question:
\begin{itemize}[leftmargin=*]
\item[]
Is disallowing after-the-fact removal necessary for achieving
subquadratic-communication BA?
\end{itemize}

\paragraph{Main result 1: disallowing ``after-the-fact'' removal is necessary.}
Our first contribution is a new lower bound showing that any (possibly randomized) 
BA protocol must incur at least $\Omega(f^2)$ communication in expectation 
in the presence of a strongly adaptive adversary 
where $f$ denotes the number of corrupt nodes.
The proof of our lower bound is inspired by  
the work of Dolev and Reischuk~\cite{consensuscclb}, 
who showed that any {\it deterministic} BA protocol 
must incur $\Omega(f^2)$ communication 
even against a {\it static} adversary.
We remark our lower bound (as well as Dolev-Reischuk) 
holds in a very strong sense: 
even when making common (possibly very strong) assumptions 
such as proof-of-work and random oracles, 
and even under a more constrained  
{\it omission} adversary who is only allowed
to omit messages sent from and to corrupt nodes, 
but does not deviate from the protocol otherwise.

\begin{theorem}[Impossibility of BA with subquadratic communication
w.r.t. a strongly adaptive adversary]
Any (possibly randomized)
BA protocol must in expectation incur at least $\Omega(f^2)$ communication
in the presence of a strongly adaptive adversary capable of performing
after-the-fact removal, 
where $f$ denotes the number of corrupt nodes.
\label{thm:lbstrongadapt}
\end{theorem}

\paragraph{Main result 2: near-optimal subquadratic BA with minimal assumptions.}
On the upper bound front, we present a subquadratic BA protocols that,
besides the necessary ``no after-the-fact removal'' assumption, 
rely only on standard cryptographic and setup assumptions. 
Our protocols achieve near-optimal resilience and expected constant rounds.

Our results improve upon existing works in two major aspects. 
Firstly, besides ``no after-the-fact removal'', 
all existing subquadratic protocols make very strong {\it additional} assumptions,
such as random oracles~\cite{algorand,aggelospos2} or proof-of-work~\cite{bitcoin}.
In particular, some works~\cite{kingsaiabarrier,algorand}
assume the ability of honest nodes to securely erase secrets from memory 
and that adaptive corruption cannot take place between when
an honest node sends a message and when it erases secrets from memory.
Such a model is referred to as the 
``erasure model'' in the cryptography literature and as ``ephemeral keys'' in Chen and Micali~\cite{algorand}.
To avoid confusing the term 
with ``after-the-fact message removal'', we rename it the {\it memory-erasure model} in this paper.
Secondly, and more importantly, even with those strong assumptions, existing protocols do not achieve the above properties (cf. Section~\ref{sec:related}).

\paragraph{The multicast model.}
In a large-scale peer-to-peer network, 
it is usually much cheaper  
for a node to multicast the same message
to everyone, than to unicast $n$ different messages (of the same length)
 to $n$ different nodes
--- even though the two have identical communication complexity in the standard pair-wise model. 
Indeed, all known consensus protocols deployed in a decentralized 
environment (e.g. Bitcoin, Ethereum) work in the multicast fashion.
Since our protocols are motivated by these large-scale peer-to-peer networks, 
we design our protocols to be multicast-based. 

A multicast-based protocol is said to have {\it multicast complexity}~$\mcc$ 
if the total number of bits multicast by all honest players is upper bounded by $\mcc$.
Clearly, a protocol with multicast complexity $\mcc$ has communication complexity $n\mcc$.
Thus, to achieve subquadratic communication complexity, 
we need to design a protocol in which only a sublinear (in $n$) number of players multicast. 

\begin{theorem}
%Assuming an adaptively secure random verifiable function \rl{elaine, check} 
Assuming standard cryptographic assumptions and
 a public-key infrastructure (PKI), 
\begin{enumerate}
\item For any constant $0<\epsilon<1/2$, 
there exists a synchronous BA protocol 
with expected $O(\compsec \cdot \poly\log(\secparam))$ multicast complexity,
expected $O(1)$ round complexity,
and $\negl(\secparam)$ error probability
that tolerates $f < (1-\epsilon) n/2$ adaptively corrupted players 
out of $n$ players in total.
\item For any constant $0<\epsilon<1/3$, 
there exists a partially synchronous BA protocol 
with expected $O(\compsec \cdot \poly\log(\secparam))$ multicast complexity,
expected $O(\Delta \cdot \poly\log(\secparam))$ time,
and $\negl(\secparam)$ error probability
that tolerates $f < (1-\epsilon) n/3$ adaptively corrupted players 
out of $n$ players in total.
\end{enumerate}
In both statements, $\kappa$ is a security parameter 
and $\compsec$ is a computational security parameter;
$\chi=\poly(\kappa)$ under standard cryptographic assumptions and 
$\chi=\poly\log(\kappa)$ if we assume sub-exponential security of the cryptographic primitives employed. 
\label{thm:introuppersync}
\end{theorem}

Our construction requires a random verifiable function (VRF) that is secure against an adaptive adversary.
Here, adaptive security means security under {\it selective opening}
of corrupt nodes' secret keys, which is a different notion
of adaptivity from in some prior works~\cite{nirvrf,brentvrf}.
Most previously known VRF constructions~\cite{VRF,nirvrf,brentvrf} do not provide security under an adaptive adversary. 
Chen and Micali~\cite{algorand} use random oracles (RO) and unique signatures to construct an adaptively secure VRF.

In the main body of the paper, we present protocols assuming an ideal functionality of VRF;
in particular, we will measure communication complexity in terms of the number of messages and state results using a statistical security parameter $\lambda$.  
In the appendix, we show how to instantiate an adaptively secure VRF under standard cryptographic assumptions such as bilinear groups,
and restate our results using $\kappa$, $\chi$, and communication
complexity in bits. 

% NOT TRUE: We stress that ours is the {\it first} partially synchronous, subquadratic-communication BA protocol.

\paragraph{Main result 3: on the necessity of setup assumptions.}
In light of the above Theorem~\ref{thm:introuppersync},
we additionally investigate 
whether the remaining setup PKI assumption is necessary.
We show that if one insists on a multicast-based protocol, indeed
some form of setup assumption is necessary for achieving sublinear multicast complexity.
Specifically, we show 
that without any setup assumption, i.e., under the plain
authenticated channels model, 
a (possibly randomized) protocol that solves BA with $\mcc$ multicast 
complexity with probability $p > 5/6$ can tolerate no more than $\mcc$ adaptive corruptions.

\begin{theorem}
[Impossibility of sublinear multicast BA without setup assumptions]
In a plain authenticated channel model without setup assumptions,
no protocol can solve BA using $\cc$ multicast complexity
with probability $p > 5/6$ under $\cc$ adaptive corruptions.
\label{thm:intromcclb}
\end{theorem}

We remark that this lower bound also applies more generally to protocols in which few nodes (i.e., less than $\mcc$ nodes) speak (multicast-style protocols are a special case).
Also note that there exist protocols with subquadratic communication and no setup assumptions that rely on many nodes to speak~\cite{kingsaiabarrier}.

%These lower-bound and upper-bound results combined make a step forward 
%in understanding the minimal assumptions needed for subquadratic-communication BA.

\paragraph{Organization.}
The rest of the paper is organized as follows.
Section~\ref{sec:related} reviews related work.
Section~\ref{sec:model} presents the model and definitions of BA.
Section~\ref{sec:lower-bound-expected} proves Theorem~\ref{thm:lbstrongadapt}.
Sections~\ref{sec:syn_roadmap},~\ref{sec:sync12},~and~\ref{sec:psync13} construct 
adaptively secure BA protocols to prove Theorem~\ref{thm:introuppersync}.
Section~\ref{sec:lower} proves Theorem~\ref{thm:intromcclb}.

%%% Local Variables:
%%% mode: latex
%%% TeX-master: "podc2019"
%%% End:

\subsection{Related Work}
\label{sec:related}

Dolev and Reischuk~\cite{consensuscclb} proved that quadratic communication is necessary for any deterministic BA protocol.
Inspired by their work, we show a similar communication complexity lower bound 
for randomized protocols, but now additionally assuming that the adversary
is strongly adaptive. 

A number of works explored randomized BA protocols~\cite{benorasync,Rabin83} to achieve expected constant round complexity~\cite{feldman1988optimal,katzbft,abraham2018synchronous} even under a strongly adaptive adversary.
A line of works~\cite{adaptivebcast,adaptivebcastrevisit,ucprobterm}
focused on a simulation-based stronger notion of adaptive security
for Byzantine Broadcast.
%, where the concern is that the adversary should not be able to observe what the sender wants to broadcast, and then adaptively corrupt the sender to flip the bit. 
These works have at least quadratic communication complexity. 

King and Saia first observed that BA can be solved with
subquadratic communication complexity if a small probability of
error is allowed~\cite{kingsaiabarrier}.
More recently Nakamoto-style protocols,
based on either proof-of-work~\cite{bitcoin} or proof-of-stake~\cite{algorand,aggelospos2}
also showed how to realize BA with subquadratic communication. 
Compared to our protocol in Section~\ref{sec:sync12}, 
these existing works make other strong assumptions,
and even with those strong assumptions, cannot simultaneously achieve near-optimal resilience and expected constant rounds.
Nakamoto consensus~\cite{bitcoin} assumes idealized proofs-of-work.
Proof-of-stake protocols assume random oracles~\cite{algorand,aggelospos2}.
King-Saia~\cite{kingsaiabarrier} and Chen-Micali~\cite{algorand} assume memory-erasure.
Nakamoto-style protocols~\cite{bitcoin,aggelospos2} and King-Saia~\cite{kingsaiabarrier} 
cannot achieve expected constant rounds.
Chen-Micali~\cite{algorand} have sub-optimal tolerance of $f<(\frac13-\epsilon)n$.

%\footnote{
%The erasure assumption may be replaced with other assumptions,
%e.g., the adversary does not see messages until the end of the round,
%or that, after corrupting a node, the adversary can start sending messages on its behalf only in the next round.} 

%%% Local Variables:
%%% mode: latex
%%% TeX-master: "podc2019"
%%% End:

\section{Model and Definition}
\label{sec:model}

\ignore{
\paragraph{Protocol execution model.}
We assume a standard protocol execution model  
with $n$ parties (also called {\it nodes}) numbered $1, 2, \ldots, n$.
An external environment denoted $\algZ$  
provides inputs to honest nodes 
and receives outputs from the honest nodes.
The adversary is denoted $\algA$,
and it can freely exchange messages with the environment $\algZ$
at any time during the execution.
We assume that all parties as well as $\algA$ and $\algZ$
are non-uniform, probabilistic polynomial-time (\ppt) Interactive Turing Machines, 
and the execution is parametrized by a security parameter $\kappa$ that is
common knowledge to all parties as well as $\algA$ and $\algZ$.
Throughout the paper, 
we would like all but a negligible in $\kappa$ fraction of executions 
to satisfy the desired consistency and validity properties.
}

\paragraph{Communication model.} 
We assume two different communication models for two different
protocols.
In Sections~\ref{sec:syn_roadmap}~and~\ref{sec:sync12}, we assume
that the network is synchronous 
and the protocol proceeds in rounds.
Every message sent by an honest node 
is guaranteed to be received by 
an honest recipient at the beginning of the next round.

In Section~\ref{sec:psync13}, we assume that the network is
partially synchronous. There are multiple ways to define partial synchrony~\cite{dworklynch}.
%In this paper, we consider the \emph{eventual synchrony} variant:
%after a global synchronization time (GST),
%the network becomes synchronous and all messages reach within $\Delta$ time;
%before GST, however, there is no guarantee on message delay.
%$\Delta$ is known to all nodes but GST is
%unknown.\kartik{revisit}
In this paper, we consider the unknown $\Delta$ variant, i.e.,
there exists a fixed message delay bound of $\Delta$ rounds but $\Delta$ is not known to any honest party.

We measure communication complexity by the number of messages sent by honest nodes.
Our protocols in
Sections~\ref{sec:syn_roadmap},~\ref{sec:sync12},~and~\ref{sec:psync13}
use multicasts only, 
that is, whenever an honest node sends a message,
it sends that message to all nodes including itself.
We say a protocol has multicast complexity $\cc$ 
if the total number of multicasts by honest nodes is bounded by $\cc$.

%The lower bound in Section~\ref{sec:lower-bound-expected} applies to protocols with any communication pattern. 

\paragraph{Adversary model.}
%Prior to the protocol execution, each node generates its public/private key pair honestly
%and sends its public key to all other nodes.
We assume a trusted PKI; every honest node knows the public
key of every other honest node.
The adversary is polynomially bounded and denoted $\algA$.
$\algA$ can \emph{adaptively} corrupt nodes any time during the protocol execution 
after the trusted setup.
The total number of corrupt nodes at the end of the execution is at most $f$.
At any time in the protocol, nodes that remain honest so far
are referred to as {\it so-far-honest} nodes and nodes that remain
honest till the end of the protocol are referred to as {\it forever-honest} nodes.
All nodes that have been corrupt are under the control of $\algA$, i.e.,
the messages they receive are forwarded to $\algA$, and $\algA$ 
controls what messages they will send in each round once they become corrupt. We assume that when a so-far-honest node $i$ multicasts a message $\msg$,
it can immediately become corrupt in the
same round and be made to send one or more messages in the same
round.
We prove our lower bound in
Section~\ref{sec:lower-bound-expected} under
 a strongly adaptive adversary that can perform an
after-the-fact removal, i.e, it can retract messages that have
already been multicast before the node becomes corrupt.
For our upper bounds in subsequent sections, we assume an
adaptive adversary can cannot perform such a retraction.
%However, the message $\msg$ that was already multicast 
%before $i$ became corrupt {\it cannot be retracted} 
%--- it will be received by all so-far-honest nodes at
%the beginning of the next round.
%In other words, the adversary is adaptive but not strongly adaptive
%(i.e., incapable of after-the-fact removal).

\ignore{
In the \emph{agreement version}, 
every node receives an input bit and 
they seek to reach consensus on a bit; 
if all honest nodes receive the same input bit $b$, 
then all honest nodes must output $b$ too.
In the \emph{broadcast version}, 
also called Byzantine broadcast, a designated
sender aims to propagate a bit to all other nodes; 
all honest nodes must output the same bit;
if the designated sender is forever-honest, 
then all honest nodes output the sender's input bit.
We formally define the two versions of the problem below. 
}

\paragraph{Agreement vs. broadcast.}
(Binary) Byzantine Agreement is typically studied in two forms.
In the \emph{broadcast version}, also called \emph{Byzantine broadcast},  
there is a \emph{designated sender} (or simply \emph{sender}) known to all nodes. 
Prior to protocol start, the sender receives an input $b \in \{0, 1\}$.
A protocol solves Byzantine broadcast with probability $p$ 
if it achieves the following properties with probability at least $p$. 
\begin{itemize}[leftmargin=8pt,topsep=4pt]
\item[-] 
{\it Termination}.
Every forever-honest node $i$ outputs a bit $b'_i$.
\item[-] 
{\it Consistency}.
If two forever-honest nodes output $b'_i$ and $b'_j$ respectively, 
then $b'_i = b'_j$.
\item[-]
{\it Validity}.
If the sender is forever-honest 
and the sender's input is $b$, then all 
forever-honest nodes output $b$.
\end{itemize}

In the \emph{agreement version}, 
sometimes referred to as Byzantine ``consensus'' in the literature,
there is no designated sender. 
Instead, each node $i$ receives an input bit $b_i \in \{0, 1\}$.
A protocol solves Byzantine agreement (BA) with probability $p$ 
if it achieves the following properties with probability at least $p$. 
\begin{itemize}[leftmargin=8pt,topsep=4pt]
\item[-] 
{\it Termination} and {\it Consistency} same as Byzantine broadcast.
\item[-]
{\it Validity}.
If all forever-honest nodes receive the same input bit $b$, 
then all forever-honest nodes output $b$. 
\end{itemize}

With synchrony and PKI, the \emph{agreement} version (where everyone receives input) 
can tolerate up to minority corruption~\cite{fitzi2002generalized}
while the broadcast version can tolerate up to $n-1$ corruptions~\cite{byzgen,dolevstrong}.
Under minority-corruption, the two versions are equivalent
from a feasibility perspective, i.e., we can construct one from the other. 
Moreover, one direction of the reduction preserves communication complexity.
Specifically, 
given an adaptively secure BA protocol (agreement version),
one can construct an adaptively secure Byzantine Broadcast protocol 
by first having the designated sender multicasting its input to everyone,
and then having everyone invoke the BA protocol.
This way, if the BA protocol has subquadratic communication complexity (resp. sublinear multicast complexity), 
so does the resulting Byzantine Broadcast protocol. 
%using the bit received from the sender as the input bit (and if 
%nothing is received or both bits are received, then use 0 as the input bit).
For this reason, we state all our upper bounds for BA and state all our lower
bounds for Byzantine Broadcast --- this makes both our upper- and lower-bounds stronger.

%%% Local Variables:
%%% mode: latex
%%% TeX-master: "podc2019"
%%% End:

\newcommand{\ExpectedComm}{\ensuremath{(\epsilon f)^2}\xspace}
\newcommand{\ExpectedCommTwoEps}{\ensuremath{\frac \epsilon 2 f^2}\xspace}

\section{Communication Lower Bound Under a Strongly Adaptive Adversary}
\label{sec:lower-bound-expected}

In this section, we prove that any (possibly randomized) BA protocol
must in expectation incur at least $\Omega(f^2)$ communication
in the presence of a strongly adaptive adversary capable of performing
after-the-fact removal.
For the reasons mentioned in Section~\ref{sec:model}, we prove our lower bound
for Byzantine Broadcast (which immediately applies to BA).
Our proof strategy builds on the classic Dolev-Reischuk lower bound~\cite[Theorem 2]{consensuscclb},
which shows that in every deterministic Byzantine 
Broadcast protocol honest nodes need to send at least $\Omega(f^2)$ messages. 

\paragraph{Warmup: the Dolev-Reischuk lower bound.}
We first explain the Dolev-Reischuk proof at a high level.
Observe that for a deterministic protocol, an execution is completely determined by the input (of the designated sender) and the adversary's strategy.
Consider the following adversary $\algA$:
$\algA$ corrupts a set $V$ of $f/2$ nodes that does not 
include the designated sender.
Let $U$ denote the set of remaining nodes.
All nodes in $V$ behave like honest nodes, except that 
(i) they ignore the first $f/2$ messages sent to them, and 
(ii) they do not send messages to each other.
Suppose the honest designated sender has input 0.
For validity to hold, all honest nodes must output $0$. 

If at most $(f/2)^2$ messages are sent to $V$ in the above execution, 
then there exists a node $p \in V$ that receives at most $f/2$ messages.
Now, define another adversary $\algA'$ almost identically as $\algA$ except that:
(i) $\algA'$ does not corrupt $p$, 
(ii) $\algA'$ corrupts all nodes in $U$ that send $p$ messages (possibly including the designated sender), prevents them from sending any messages to $p$, but behaves honestly to other nodes.
Since $p$ receives at most $f/2$ messages under $\algA$, 
$\algA'$ corrupts at most $f$ nodes.

Observe that honest nodes in $U$ receive identical messages from all other nodes in the two executions. 
So these nodes still output 0 under $\algA'$.
However, $p$ does not receive any message but has to output some value. 
If this value is $1$, consistency is violated.
If $p$ outputs $0$ when receiving no messages, we can let the sender 
send $1$ under $\algA$ and derive a consistency violation under $\algA'$ 
following a symmetric argument. 

\paragraph{Our lower bound.}
We now extend the above proof to randomized protocols. 
In a randomized protocol, there are two sources of randomness
that need to be considered carefully. 
On one hand, honest nodes can use randomization to
their advantage. 
On the other hand, an adaptive adversary can also leverage randomness. 
Indeed our lower bound uses a randomized adversarial strategy.
In addition, our lower bound crucially relies on the adversary being \emph{strongly adaptive} --
the adversary can observe that a message is sent by an honest node 
$h$ to any other party in a given round $r$, decide to adaptively corrupt $h$,
and then remove messages sent by $h$ in round $r$.
We prove the following theorem --- here we say that a protocol solves
Byzantine Broadcast with probability $q$ iff 
for any non-uniform p.p.t.\ strongly adaptive
adversary, with probability $q$, every honest  
node outputs a bit at the end of the protocol, 
and consistency and validity are satisfied. 

\begin{theorem}
If a protocol solves Byzantine 
Broadcast with $\frac34 + \epsilon$ probability 
against a strongly adaptive adversary, 
then in expectation, honest nodes collectively need to send at
least $\ExpectedComm$ messages. 
\end{theorem}

\begin{proof}
For the sake of contradiction,
suppose that a protocol solves Byzantine Broadcast against a strongly adaptive adversary with $\frac 34 + \epsilon$ probability using less than $\ExpectedComm$ expected messages.
This means, regardless of what the adversary does, the protocol errs 
(i.e., violate either consistency, validity or termination) 
with no more than $\frac 14 - \epsilon$ probability.
We will construct an adversary that makes the protocol err with a probability larger than the above.  

Without loss of generality, 
assume that there exist $\lceil n/2 \rceil$ nodes that output 0
with at most 1/2 probability if they receive no
messages. 
(Otherwise, then there must exist $\lceil n/2 \rceil$ nodes that output 1 with at most 1/2 probability if they receive no messages, and the entire proof follows from a symmetric argument.)
Let $V$ be a set of $f/2$ such nodes not containing the designated sender.
Note that these nodes may output 1 or they may simply not terminate if they receive no messages.
(We can always find such a $V$ because $f/2 < \lceil n/2 \rceil$).
Let $U$ denote the remaining nodes.
Let the designated sender send 0.

Next, consider the following adversary $\algA$ that corrupts $V$ and makes nodes in $V$ behave honestly except that:
\begin{enumerate}[leftmargin=*]
\item Nodes in $V$ do not send messages to each other.
\item Each node in $V$ \emph{ignores} (i.e., pretends that it does not receive) the first $f/2$ messages sent to it by nodes in $U$.
\end{enumerate}

For a protocol to have an expected communication complexity of $\ExpectedComm$, honest nodes 
collectively need to send fewer than that many messages in expectation \emph{regardless of} the adversary's strategy.
Let $z$ be a random variable denoting the number of messages sent by honest
nodes to $V$. 
We have $E[z] < \ExpectedComm$.
Let $X_1$ be the event that $z \leq \ExpectedCommTwoEps$.
By Markov's inequality, $\Pr[z > \frac 1{2\epsilon} E[z]] < 2\epsilon$. 
Thus, $\Pr[z \leq \ExpectedCommTwoEps] \geq \Pr[z \leq \frac
1{2\epsilon} E[z]] > 1 - 2\epsilon$.

Let $X_2$ be the event that among the first $\ExpectedCommTwoEps$ messages, 
a node $p$ picked uniformly at random from $V$ by the adversary receives at most $f/2$ messages.
Observe that among the first $\ExpectedCommTwoEps = 2\epsilon|V|(f/2)$ messages,
there exist at most $2\epsilon|V|$ nodes that receive more than $f/2$ of those.
Since $p$ has been picked uniformly at random from $V$, $\Pr[X_2] \geq 1-2\epsilon$.
Thus, we have that 
\begin{align*}
\Pr[X_1 \cap X_2] &= \Pr[X_1] + \Pr[X_2] - \Pr[X_1 \cup X_2]  \\
			  &> (1-2\epsilon) + (1-2\epsilon) - 1 
              = 1-4\epsilon.
\end{align*}

Now, define another adversary $\algA'$ almost identically as $\algA$ except that: 
\begin{enumerate}[leftmargin=*]
\item $\algA'$ 
picks a node $p \in V$ uniformly at random and corrupts
everyone else in $V$ except $p$. 
\item $\algA'$ blocks the first $f/2$ attempts that nodes in $U$ send $p$ messages.
In other words, whenever some node $s \in U$ attempts to send a message to $p$ in a round, 
if this is within the first $f/2$ attempts that nodes in $U$ send $p$ messages,
$\algA'$ immediately corrupts $s$ (unless $s$ is already corrupted) and removes the message $s$ sends $p$ in that round. 
Corrupted nodes in $U$ behave honestly otherwise.
(In particular, after the first $f/2$ messages from $U$ to $p$ have been blocked,
corrupted nodes behave honestly to $p$ as well.)
\end{enumerate}

Observe that $X_1 \cap X_2$ denotes the event that under adversary $\algA$, 
the total number of messages sent by honest nodes to $V$ is less than $\ExpectedCommTwoEps$ 
and among those, the randomly picked node $p$ has received at most $f/2$ messages.
In this case, $p$ receives no message at all under adversary $\algA'$.
By the definition of $V$, $p$ outputs 0 with at most 1/2 probability if it receives no messages.
Let $Y_1$ be the event that $p$ does \emph{not} output 0 under $\algA'$.
Recall that $Y_1$ includes the event that $p$ outputs 1 as well as the event that $p$ does not terminate. 
We have $\Pr[Y_1] \geq \Pr[Y_1 | X_1 \cap X_2] \cdot \Pr[X_1 \cap X_2] > \frac 12 (1-4\epsilon)$.

Meanwhile, we argue that honest nodes in $U$ cannot distinguish $\algA$ and $\algA'$.
This is because the only difference between the two scenarios is that,
under $\algA$, the first $f/2$ messages from $U$ to $p$ are intentionally ignored by a corrupt node $p$,
and under $\algA'$, the first $f/2$ messages from $U$ to an honest $p$ are blocked by $\algA'$ using after-the-fact removal.
Thus, honest nodes in $U$ receive identical messages under $\algA$ and $\algA'$ and cannot distinguish the two adversaries.
Under $\algA$, they need to output 0 to preserve validity.
Recall that the protocol solves Byzantine broadcast with at least $\frac 34 + \epsilon$ probability.
Thus, with at least the above probability, all honest nodes in $U$ output 0 under $\algA$. 
Let $Y_2$ be the event that all honest nodes in $U$ output 0 under $\algA'$.
Since they cannot distinguish $\algA$ and $\algA'$, $\Pr[Y_2] \geq \frac 34 + \epsilon$. 

If $Y_1$ and $Y_2$ both occur, then the protocol errs under $\algA'$: either consistency or termination is violated. We have
\begin{align*}
\Pr[Y_1 \cap Y_2] &= \Pr[Y_1] + \Pr[Y_2] - \Pr[Y_1 \cup Y_2]  \\
			  &> \frac 12 (1-4\epsilon) + (\frac 34 + \epsilon) - 1 
              = \frac 14 - \epsilon.
\end{align*}
This contradicts the hypothesis that the protocol solves Byzantine broadcast with $\frac 34 + \epsilon$ probability.
\end{proof}

%%% Local Variables:
%%% mode: latex
%%% TeX-master: "podc2019"
%%% End:

\section{Subquadratic BA under Synchrony: $f<(1/3-\epsilon)n$}
\label{sec:syn_roadmap}

This section presents the main ingredients 
for achieving subquadratic BA.
In this section, we opt for conceptual simplicity over other desired properties.
In particular, the protocol in this section 
tolerates only $\frac{1}{3}-\epsilon$ fraction of adaptive corruptions, 
and completes in $O(\numrounds)$ rounds. 
In the next section, we will show how to improve the 
resilience to $\frac{1}{2}-\epsilon$ and round complexity to expected $O(1)$.

\ignore{
We note that this simple protocol (that suffers in resilience) 
advances the state-of-the-art: 
to the best of our knowledge,
no protocol achieve sublinear multicast complexity 
in the presence of any constant fraction
of adaptive corruptions.
}

\subsection{Warmup: A Simple Quadratic 
BA Tolerating 1/3 Corruptions}
We first describe an extremely simple quadratic 
BA protocol, inspired by the Phase-King paradigm~\cite{dcbook}, that tolerates
less than 1/3 corruptions. 
The protocol proceeds in $\numrounds$ iterations $r = 1, 2, \ldots \numrounds$,
and every iteration consists of two rounds.
For the time being, assume a random leader election oracle that
elects and announces a random leader at the beginning of every iteration.
At initialization, every node $i$ sets $b_i$ 
to its input bit, and sets its ``sticky flag'' $F = 1$ (think of the
sticky flag as indicating whether to ``stick'' to the bit in the
previous iteration).
Each iteration $r$ now proceeds as follows where
all messages are signed, and only messages
with valid signatures are processed:
\begin{enumerate}[leftmargin=5mm]
\item 
The leader of iteration $r$ flips a random coin $b$ and \underline{multicasts}
$(\Propose, r, b)$.
Every node $i$ sets $b^*_i := b_i$ if $F = 1$ or
if it has not heard a valid proposal from the current iteration's leader.
Else, it sets $b^*_i := b$
where $b$ is the proposal 
heard from the current iteration's leader (if proposals for both $b = 0$
and $b=1$ 
have been observed, choose an arbitrary bit).

\item
Every node $i$ \underline{multicasts} $(\Vote, r, b^*_i)$.
If at least $\frac{2n}{3}$ votes from distinct nodes have been received and vouch
for the same $b^*$, set $b_i := b^*$ and $F := 1$;
else, set $F := 0$.
\end{enumerate}
At the end of the last iteration, each node outputs the bit that it last voted for. 

% (0 if it has never sent a vote). IMPOSSIBLE

In short, in every iteration, every node
either switches to the leader's proposal (if any has been observed)
or it sticks to its previous ``belief'' $b_i$.
This simple protocol works because of the following observations. 
Henceforth, we refer to a collection of $\frac{2n}{3}$ votes from distinct nodes 
for the same iteration and the same $b$ as a \emph{certificate} for $b$.
\begin{itemize}[leftmargin=*]
\item[-] 
{\it Consistency within an iteration.}
Suppose that in iteration $r$, honest node $i$ 
observes a certificate for $b$ from a set of nodes denoted $S$,
and honest node $j$ observes a certificate for $b'$ from a set $S'$.
By a standard quorum intersection argument,
$S \cap S'$ must contain at least one forever-honest node.
%(recall that a forever-honest node remains honest throughout the entire execution).
Since honest nodes vote uniquely, it must be that $b = b'$.

\item[-] 
{\it A good iteration exists.}
Next, suppose that in some iteration $r$ the leader 
is honest. We say that this leader chooses a lucky bit
$b^*$ iff 
in iteration $r-1$, no honest node has seen a certificate for $1-b^*$.
This means, in iteration $r$, every honest node either 
sticks with $b^*$ or switches to the leader's proposal of $b^*$.
Clearly, an honest leader chooses a lucky $b^*$ with probability
at least $1/2$. Except with $\exp(-\Omega(\numrounds))$ probability,
an honest-leader iteration with a lucky choice exists.

\item[-] 
{\it Persistence of honest choice after a good iteration.}
Now, as soon as we reach an iteration (denoted $r$) with an honest leader
and its choice of bit $b^*$ is lucky, 
then all honest nodes will vote for $b^*$ in iteration $r$. 
Thus all honest nodes will hear certificates for $b^*$ in 
iteration $r$; therefore, they will all stick to $b^*$ in 
iteration $r+1$. By induction, 
in all future iterations they will stick to $b^*$.

\item[-] 
{\it Validity.}
If all honest nodes receive the same bit $b^*$ as input
then due to the same argument as above the bit $b^*$ will always stick
around in all iterations.
\end{itemize}

\subsection{Subquadratic Communication through Vote-Specific Eligibility}
\label{sec:subquad}

The above simple protocol requires in expectation linear number of multicast messages (in each round
every node multicasts a message). 
We now consider how to improve the multicast complexity of the warmup protocol. 
We will also remove the idealized leader election oracle in the process. 

\paragraph{Background on VRFs.}
We rely on a verifiable random function (VRF)~\cite{VRF}.
A trusted setup phase is used to generate a public-key infrastructure (PKI): 
each node $i \in [n]$ obtains a VRF secret key $\sk_i$, 
and its corresponding public key $\pk_i$.
A VRF evaluation on the message $\mu$ 
denoted $(\rho, \pi) \leftarrow {\sf VRF}_{\sk_i}(\mu)$ 
generates a deterministic pseudorandom value $\rho$ 
and a proof $\pi$ such that $\rho$ is computationally indistinguishable
from random without the secret key $\sk_i$, and with $\pk_i$ everyone
can verify from the proof $\pi$ that $\rho$ is evaluated correctly.  
We use ${\sf VRF}^1$ to denote the first output (i.e., $\rho$ above) of the VRF.

\paragraph{Strawman: the Chen-Micali approach.}
We first describe the paradigm of Chen and Micali~\cite{algorand} 
but we explain it in the context of our warmup protocol.
Imagine that now not everyone is required to vote in a round $r$. 
Instead, we use the function ${\sf VRF}^1_{\sk_i}(\Vote, r) < D$
to determine whether $i$ is eligible to vote in round $r$ where 
$D$ is a difficulty parameter appropriately chosen
such that, in expectation, $\comsize$ many nodes 
would be chosen to vote in each round. 
When node $i$ sends a \Vote message, 
it attaches the VRF's evaluation outcome 
as well as the proof such that 
every node can verify its eligibility using its public key $\pk_i$.
Correspondingly, when we tally votes, the original threshold $\frac{2n}3$ should be changed to $\frac{2\comsize}3$, 
i.e, two-thirds of the expected committee size.

Evaluating the VRF requires knowing the node's secret key. 
Thus, only the node itself knows at what rounds it 
is eligible to vote. 
This may seem to solve the problem because the adversary 
cannot predict in advance who will be sending messages in every round.
The problem with this is that 
once an adaptive adversary $\algA$
notices that some player $i$ was eligible to vote for $b$ in round $r$
(because $i$ just sent a valid vote for $b$), $\algA$ can corrupt $i$
immediately and make $i$ vote for $1-b$ in the same round!

To tackle this precise issue, 
Chen and Micali~\cite{algorand} relies on the memory-erasure model (referred to as ephemeral keys in their paper) 
and a {\it forward-secure} signing scheme.
Informally, in a forward secure signing scheme,
in the beginning, a node has a key that can sign any messages from any round;
after signing a message for round $t$, the node updates its key
to one that can henceforth sign only messages for round $t+1$ or higher, 
and the round-$t$ secret key should be immediately erased at this point. 
This way, even if the attacker instantly corrupts a node, it cannot cast
another vote in the same round.

\paragraph{Our key insight: bit-specific eligibility.}
Our key insight is to make the eligibility bit-specific.
To elaborate, the committee eligible
to vote for $b$ in round $r$ is chosen 
{\it independently} from the committee eligible to vote on $1-b$ in the same round. 
Concretely, 
node $i$ is eligible to send a \Vote message for 
the bit $b \in \{0, 1\}$ in round $r$ iff
${\sf VRF}^1_{\sk_i}(\Vote, r, {\color{blue} b}) < D$,
where $D$ is the aforementioned difficulty parameter.

What does this achieve? Suppose that the attacker sees some node $i$ votes for the bit $b$ in round $r$.
Although the attacker can now immediately corrupt $i$, 
%the only thing this can definitely guarantee is to make player~$i$ vote again for the same bit $b$ in round $r$,
%but the honest player has already done that! 
the fact that $i$ was allowed to vote for $b$ in round $r$ 
does not make $i$ any more likely to be eligible to vote for $1-b$ in the same round.
Thus, corrupting $i$ is
no more useful to the adversary than corrupting any other node.

Finally, since we already make use of the VRF, as a by-product
we can remove the idealized leader election oracle
in the warmup protocol: 
a node $i$ is eligible for making a proposal in iteration $r$
iff ${\sf VRF}^1_{\sk_i}(\Propose, r, b) < D_0$ where $D_0$ is
a separate difficulty parameter explained below. 
Naturally, the node attaches the VRF evaluation outcome and proof 
with its proposal so that others can verify its eligibility.

\paragraph{Difficulty parameters.}
The two difficulty parameters $D$ and $D_0$ need to be specified differently.
Recall that $D$ is used to elect a committee in each round for sending
\Vote messages; 
and $D_0$ is used for leader election.
\begin{enumerate}
\item 
$D$ should be set such that each committee is $\comsize$-sized
in expectation; whereas 
\item $D_0$ should be set such that
every node has a $\frac 1{2n}$ probability to be eligible to propose.  
\end{enumerate}
Since we are interested in making communication scale better with $n$, 
we assume $n > \comsize$;
otherwise, one should simply use the quadratic protocol.

\paragraph{Putting it together.}
More formally, we use the phrase ``node $i$ {\it conditionally multicasts}
a message $(\MType, r, b)$'' 
to mean that node $i$ checks if it is eligible to vote for $b$ in
iteration $r$ and if so, it multicasts
$(\MType, r, b, i, \pi)$,
where $\MType \in \{\Propose, \Vote\}$ stands for the type of the message
and $\pi$ is a proof proving that $i$ indeed is eligible 
(note that $\pi$ includes
both the pseudorandom evaluation result and the proof output by the VRF). 
Now, our new committee-sampling based subquadratic protocol
is almost identical to the warmup protocol except for the following changes:
\begin{itemize}[leftmargin=5mm]
\item 
every occurrence of \underline{multicast}
is now replaced with  ``\underline{conditionally multicast}''; 
\item 
the threshold of certificates (i.e., number of votes for a bit to stick)  
is now $\frac{2\comsize}{3}$;
and 
\item upon receiving every message, 
a node checks the proof 
to verify the sender's eligibility to send that message.
\end{itemize}

\subsection{Proof Sketch}
To help our analysis, we shall abstract away
the cryptography needed for
eligibility election, and instead think of eligibility 
election as making queries to a trusted party called \Fmine.
We call an attempt for node $i$ to check its eligibility to send
either a \Propose or \Vote message   
a {\it mining} attempt for a \Propose or \Vote message (inspired
by Bitcoin's terminology where miners ``mine'' blocks).
Specifically, if a node~$i$ wants to check its eligibility 
for sending $(\MType, r, b)$
where $\MType \in \{\Propose, \Vote\}$,
it calls $\Fmine.{\tt mine}(\MType, r, b)$, and \Fmine
shall flip a random coin with appropriate probability to determine
whether this ``mining'' attempt is successful. 
If successful, $\Fmine.{\tt verify}((\MType, r, b), i)$
can vouch to any node of the successful attempt
-- this is used in place of verifying the VRF proof. 
If a so-far-honest node makes a mining attempt for some $(\MType, r, b)$,
it is called an {\it honest mining attempt} (even if the node immediately
becomes corrupt afterwards in the same round).
Else, if an already corrupt node 
makes a mining attempt, 
it is called a {\it corrupt mining attempt}.

We now explain why our new protocol works by following
similar arguments as the underlying BA --- but now we must
additionally analyze the stochastic process induced by eligibility election.

\paragraph{Consistency within an iteration.}
We first argue why ``consistency within an iteration'' still holds
with the new protocol.
There are at most $(\frac{1}{3} - \epsilon) n$ corrupt nodes,
each of which might try to mine for two votes (one for each bit) 
in every iteration~$r$.
On the other hand, each so-far-honest node
will try to mine for only one vote in each iteration.
Therefore, in iteration~$r$, 
the total number of mining attempts (honest and corrupt) 
for \Vote messages is at most $2(\frac{1}{3} - \epsilon)n + 
(\frac{2}{3} + \epsilon)n = (\frac{4}{3} - \epsilon) n$,
each of which is \textbf{independently} successful with probability $\frac{\comsize}{n}$.
Hence, if there are $\frac{2 \comsize}{3}$ votes
for each of the bits $0$ and $1$,
this means there are at least in total 
$\frac{4 \comsize}{3}$ successful mining attempts,
which happens with $\exp(-\Omega(\comsize))$ probability,
by the Chernoff bound.
Therefore, except with $\exp(-\Omega(\comsize))$ probability,
if any node sees $\frac{2 \comsize}{3}$ votes for some bit $b$, then
no other node sees $\frac{2 \comsize}{3}$ votes for a different bit $b'$.

\paragraph{A good iteration exists.}
We now argue why ``a good iteration exists'' in our new protocol.
Here, for an iteration  $r$ to be good, the following must hold: 1) 
a \emph{single} so-far-honest node 
successfully mines a $\Propose$ message, 
and no already corrupt node successfully mines a $\Propose$ message;
and 
2) if some honest nodes want to stick to 
a bit $b^*$ in iteration $r$, the leader's random coin must agree with $b^*$.
(Note that if multiple so-far-honest nodes successfully mine $\Propose$ messages, 
this iteration is not a good iteration). 
Every so-far-honest node makes only one \Propose mining
attempt per iteration. 
Every already corrupt node can make two \Propose mining attempts in an iteration,
one for each bit.
Since our \Propose mining difficulty parameter $D_0$ is set 
such that each attempt succeeds with $\frac 1{2n}$ probability,  
in every iteration, with $\Theta(1)$ probability, 
a single honest \Propose mining attempt 
is successful and no corrupt \Propose mining attempt is successful.
Since our protocol consists of $\comsize$ iterations, 
a good iteration exists except with $\exp(-\Omega(\comsize))$ probability,

\paragraph{Remainder of the proof}.
Finally, ``persistence of honest choice after a good iteration'' and ``validity'' hold 
in a relatively straightforward fashion by applying the standard Chernoff bound.

\paragraph{Remark.}
We stress that for the above argument to hold, 
it is important that the eligibility be 
tied to the bit being proposed/voted.
Had it not been the case, the adversary could observe whenever
an honest node sends $(\MType, r, b)$, 
and immediately
corrupt the node in the same round and make it send $(\MType, r, 1-b)$ too.
If $\MType$ is $\Vote$,
whenever there are $\frac{2\comsize}{3}$ votes for $b$
in iteration $r$, by corrupting all these 
nodes that are eligible to vote, the adversary can construct $\frac{2\comsize}{3}$ votes 
for $1-b$, and thus ``consistency within an iteration'' does not hold. 
If $\MType$ is $\Propose$,
whenever there is a so-far-honest leader in iteration $r$,
by corrupting this leader, the adversary gets a corrupt leader,
and thus no good iteration would exist.

%%% Local Variables:
%%% mode: latex
%%% TeX-master: "podc2019"
%%% End:

\section{Subquadratic BA under Synchrony: $f<(1/2-\epsilon)n$}
\label{sec:sync12}

In this section, we present our synchronous BA protocol
that achieves expected subquadratic communication complexity 
(expected sublinear multicast complexity)
and expected constant round complexity,
and tolerates $f < (\frac 12 -\epsilon) n$ adaptive corruptions.
Our starting point is Abraham et al.~\cite{abraham2018synchronous},
a synchronous quadratic BA protocol tolerating $f < n/2$ corruptions.
We explain Abraham et al. at a high level and 
then apply the techniques introduced in the previous section 
to achieve subquadratic communication complexity. 

\subsection{Warmup: Quadratic BA Tolerating 1/2 Corruptions}
\label{sec:sync12:warmup}
Our description below assumes $n=2f+1$ nodes in total.
The protocol runs in iterations $r=1,2,\ldots$
Each iteration has four synchronous rounds
called $\Status$, $\Propose$, $\Vote$, and $\Commit$, respectively.
Messages sent at the beginning of a round will be received before next round.
All messages are signed.
Henceforth, a collection of $f+1$ (signed) iteration-$r$ $\Vote$ messages
for the same bit $b \in \{0, 1\}$ from
distinct nodes is said to be an {\it iteration-$r$ certificate}
for $b$. A certificate from a higher iteration is said to be a
higher certificate.
For the time being, assume a random leader election oracle that
elects a random leader $L_r$ at the beginning of every iteration $r$.

Below is the protocol for an iteration $r \geq 2$.
The protocol for the very first iteration $r=1$ skips the $\Status$ and $\Propose$ rounds. 
\begin{enumerate}[leftmargin=5mm,itemsep=1pt]
\item $\Status$. 
Every node multicasts a $\Status$ message 
of the form $(\Status, r, b, \Cert)$
containing the highest certified bit $b$ it has seen so far
as well as the corresponding certificate $\Cert$.

\item $\Propose$. The leader $L_r$ chooses a bit $b$ with a 
highest certificate denoted $\Cert$ breaking ties arbitrarily.
The leader multicasts $(\Propose, r, b, \Cert)$. 
To unify the presentation, we say that a bit $b$ without any 
certificate has an iteration-0 certificate and it is treated as the lowest ranked certificate. 

\item $\Vote$. 
For the very first iteration $r=1$, a node votes for its input bit $b$ by multicasting $(\Vote, r=1, b)$. 

For all iterations $r\geq2$,
if a validly signed $(\Propose, r, b, \Cert)$ message
has been received from $L_r$ with a certificate $\Cert$ for $b$,  
and if the node has not 
observed a \emph{strictly} higher
certificate for $1-b$, 
it multicasts an iteration-$r$ $\Vote$ message for $b$ 
of the form $(\Vote, r, b)$
with the leader's proposal attached.
%\footnote{The leader's proposal will not be included in a $\Vote$ message or a certificate.}.
Importantly, if the node has observed a certificate for the opposite bit $1-b$ from the same iteration as \Cert, it \emph{will} vote for $b$.

\item $\Commit$.
If a node has received $f+1$ iteration-$r$ signed votes 
for the same bit $b$ from distinct nodes
(which form an iteration-$r$ certificate $\Cert$ for $b$) 
and no iteration-$r$ vote for $1-b$, 
it multicasts an iteration-$r$ $\Commit$ message for $b$ of the form $(\Commit, r, b)$ 
with the certificate \Cert attached.
 
\item[$\star$] 
(This step is not part of the iteration and can be executed at any time.)
If a node has received $f+1$ $\Commit$ messages for the same $b$ from the same iteration from distinct nodes, 
it multicasts a termination message of the form $(\Terminate, b)$ with the $f+1$ $\Commit$ messages attached. 
The node then outputs $b$ and terminates. 
This last message will make all other honest nodes multicast the same $\Terminate$ message, output $b$ and terminate in the next round.  
\end{enumerate}

\paragraph{Consistency.}
The protocol achieves consistency due to the following key property.
\emph{If an honest node outputs a bit $b$ in iteration~$r$, 
then no certificate for $1-b$ can be formed in iteration $r$ and all subsequent iterations.}
We explain why this property holds below.

An honest node outputs $b$ in iteration $r$, 
only if it has observed $f+1$ iteration-$r$ $\Commit$ messages (from distinct nodes) for $b$.
One of these must have been sent by an honest node henceforth indexed by~$i^*$.  
For an iteration-$r$ certificate for $1-b$ to exist, 
an honest node must have multicast a vote for $1-b$.
But in that case, $i^*$ would have received this conflicting vote i
and thus would not have sent the commit message for $b$.
We have reached a contradiction.
Thus, we can rule out any iteration-$r$ certificate for $1-b$.

Furthermore, by the end of iteration $r$, 
all nodes will receive from node $i^*$ an iteration-$r$ certificate for $b$.
Since no iteration-$r$ certificate for $1-b$ exists, no honest node votes for $1-b$ in iteration $r+1$;
hence, no iteration-$(r+1)$ certificate for $1-b$ can come into existence;
hence no honest node votes for $1-b$ in iteration $r+2$, and so on.
The preference for a higher certificate ensures consistency for all subsequent iterations 
following a simple induction.

\paragraph{Validity.}
Recall that the first iteration skips $\Status$ and $\Propose$ and directly starts with $\Vote$.
If all honest nodes have the same input bit $b$, 
then they all vote for $b$ in the first iteration.
By the end of the first iteration, 
every honest node has an iteration-1 certificate for $b$
and no iteration-1 certificate for $1-b$ exists.
Validity then follows from consistency.

\paragraph{Expected constant round complexity.}
Once an iteration has an honest leader,
it will sign a unique proposal for the bit $b$ with the highest certificate reported by honest nodes. 
Then, all honest nodes send $\Vote$ and $\Commit$ messages for $b$, 
output and terminate in that iteration.
Since leaders are selected at random, in expectation, an honest leader emerges in two iterations.

\subsection{Subquadratic Communication through Vote-Specific Eligibility}
\label{sec:sync12:subquad}
The above simple protocol requires quadratic communication (in each round every node multicasts a message). 
We now improve the communication complexity to subquadratic 
and we will also remove the idealized leader election oracle in the process. 

We now use the vote-specific eligibility
to determine for each iteration, 
who is eligible for sending $\Status$, $\Propose$, $\Vote$ and $\Commit$ messages for 0 and 1, respectively.
To keep the presentation simple, we abstract away the cryptographic primitives for eligibility election and model it as an ideal functionality \Fmine. 
As before, we call an attempt for node $i$ to check eligibility to send a message a {\it mining} attempt. 
%This is inspired by Bitcoin's terminology where miners ``mine'' blocks.
Concretely, node $i$ is eligible to send $(\MType, r, b)$ 
where \MType is \Status, \Vote, or \Commit, iff
\[ \Fmine.\mine(i, \MType, r, b) < D, \]
node $i$ is eligible to send $(\Terminate, b)$ iff 
\[ \Fmine.\mine(i, \Terminate, b) < D, \] 
and node $i$ is eligible to send $(\Propose, r, b)$ iff 
\[ \Fmine.\mine(i, \Propose, r, b) < D_0. \]
$D$ and $D_0$ are appropriate {\it difficulty} parameters 
such each mining attempt for \Status/\Vote/\Commit/\Terminate has a $\comsize/n$ probability to be successful and each mining attempt for leader proposal has a $1/n$ probability to be successful. 
As before, we assume $n > \comsize$;
otherwise, one should simply use the quadratic protocol.

We use the phrase ``node $i$ {\it conditionally multicasts} a message'' 
to mean that node $i$ checks with \Fmine if it is eligible to send that message 
and only multicasts the message if it is.
Now, the committee-sampling based subquadratic protocol
is almost identical to the warmup protocol except for the following changes:
\begin{itemize}[leftmargin=5mm]
\item 
every occurrence of \underline{multicast}
is now replaced with ``\underline{conditionally multicast}''; 

\item 
every occurrence of \underline{$f+1$} $\Vote$ or $\Commit$ messages 
is now replaced with \underline{$\comsize/2$} messages of that type;
and

\item upon receiving a message of the form $(i, \msg)$ 
(including messages attached with other messages), 
a node invokes $\Fmine.\ver(i, \msg)$ to verify node $i$'s eligibility to send that message.
Note that $\msg$ can be of the form $(\MType, r, b)$
where $\MType \in \{\Status, \Propose, \Vote, \Commit\}$
or of the form $(\Terminate, b)$.
\end{itemize}

\subsection{Proof}
\label{sec:sync12:proof}

We prove our new protocol works in this subsection. 
The proofs mostly follow the sketch in Section~\ref{sec:sync12:warmup}
--- except that we now need to analyze the stochastic process induced by eligibility.
%Our stochastic analysis here is performed assuming an idealized \Fmine,  
%and this idealized oracle will be removed later in Section~\ref{sec:real}.

To prove consistency and validity, we first establish the following lemma.  
%To show that each bad event we care about happens with $\exp(-\Omega(\comsize))$ probability.
%We will then show there are at most $\poly(\comsize)$
%such bad events that we need to take a union bound over,
%so the overall error probability is still exponentially small in~$\comsize$. 
%Recall that $n>\comsize$ and that the adversary can make at most $(1/2-\epsilon)n$ adaptive corruptions where $0 < \epsilon < 1/2$ is a constant.

\begin{lemma}
Except for $\exp(-\Omega(\comsize))$ probability,
for any \Status/\Vote/\Commit message 
for bit $b$ in iteration $r$,
less than $\comsize/2$ eventually-corrupt nodes are eligible to send it.
% and (ii) if $(\frac12 + \frac\epsilon2)n$ so-far-honest nodes attempt to mine it, then at least $\comsize/2$ so-far-honest nodes send it. 
\label{lemma:chernoff}
\end{lemma}
\begin{proof}
%For (i), observe that there are at most $(\frac12-\epsilon)n$ eventually-corrupt nodes.
%By our choice of $D$, each eventually-corrupt node independently has a $\comsize/n$ probability to be eligible to send the said message.
%A simple Chernoff bound completes the proof.
Recall that the adversary can make at most $(1/2-\epsilon)n$ adaptive corruptions where $0 < \epsilon < 1/2$ is a constant.
By our choice of $D$, each so-far-honest node independently has a $\comsize/n$ probability to send the said message. 
The lemma follows from a simple Chernoff bound. 
\end{proof}

\begin{theorem}[Consistency]
Except for $\exp(-\Omega(\comsize))$ probability,
if an honest node outputs a bit $b$ in iteration~$r$, 
then no certificate for $1-b$ can be formed in iteration $r$ and all subsequent iterations.
\label{thm:safety}
\end{theorem}

\begin{proof}
An honest node outputs $b$ in iteration $r$, 
only if it has observed $\comsize/2$ $\Commit$ messages for $b$.
By Lemma~\ref{lemma:chernoff}, except for $\exp(-\Omega(\comsize))$ probability, 
not all of them are sent by eventually-corrupt nodes;
in other words, one of the $\Commit$ messages was sent by a forever-honest node henceforth indexed by~$i^*$.  
Similarly, for an iteration-$r$ certificate for $1-b$ to exist,
except for $\exp(-\Omega(\comsize))$ probability, 
one forever-honest node has multicast a vote for $1-b$.
But in that case, $i^*$ would have received this conflicting vote and would not have sent the $\Commit$ message for $b$.
We have reached a contradiction.
Thus, no iteration-$r$ certificate for $1-b$ exists except 
for $\exp(-\Omega(\comsize))$ probability.

Furthermore, by the beginning of iteration $r+1$, 
all forever-honest nodes will receive from node $i^*$ an iteration-$r$ certificate for $b$.
The lack of iteration-$r$ certificate for $1-b$ together with the preference to higher certificate ensures that no forever-honest node will vote for $1-b$ in iteration $r+1$.
To form a certificate for $1-b$ in a subsequent iteration,
all $\comsize/2$ votes have to come from eventually-corrupt nodes, 
which happens with $\exp(-\Omega(\comsize))$ probability by Lemma~\ref{lemma:chernoff}.
An induction then completes the proof. 
\end{proof}

\begin{theorem}[Validity]
Except for $\exp(-\Omega(\comsize))$ probability, 
if all honest nodes have the same input bit $b$,
then all nodes will output $b$.
\label{thm:validity}
\end{theorem}

\begin{proof}
Straightforward from Lemma~\ref{lemma:chernoff}: 
in the first iteration, except for the said probability,
there will be sufficient forever-honest nodes to send $(\Vote, r=1, b)$,
and there will not be sufficient eventually-corrupt nodes to vote for $1-b$.
Validity then follows from consistency.
\end{proof}

We then turn to analyze round complexity and communication/multicast complexity.
We say an iteration $r$ is a good iteration if  
a \emph{single} so-far-honest node successfully mines a \Propose message, 
and no already-corrupt node successfully mines a \Propose message.
(Note that if multiple so-far-honest nodes successfully mine \Propose messages, 
this iteration is not a good iteration). 

\begin{lemma}
If fewer than $\epsilon n/2$ forever-honest nodes have terminated,
then, every iteration independently has a $\Theta(1)$ probability to be a good iteration.
\label{lemma:honest-unique-leader}
\end{lemma}
\begin{proof}
In any fixed iteration $r$, suppose there are $n_h$ so-far-honest nodes that have not terminated and $n_c$ already-corrupt nodes.  
Each so-far-honest node makes one attempt to propose (either 0 or 1),
whereas each already-corrupt node can make two attempts to propose (both 0 and 1).
Recall that we set $D_0$ such that each mining attempt for $\Propose$ succeeds with probability $1/n$.

The probability that exactly one honest $\Propose$ attempt succeeds and no corrupt $\Propose$ attempt succeeds is   
${n_h \choose 1} \frac{1}{n} (1-\frac1n)^{n_h-1+2n_c}$.
Observe that $n_h+n_c<n$, $n_c<n/2$ and $n_h>n/2$,
the above expression is greater than $\frac12 \cdot (1-\frac1n)^{1.5n} = \Theta(1)$.
\end{proof}

\begin{lemma}
Except for $\exp(-\Omega(\comsize))$ probability, 
if at least $\epsilon n/2$ forever-honest nodes have terminated,
all so-far-honest nodes terminate by the end of the next round.
\label{lemma:terminate}
\end{lemma}

\begin{proof}
Each of the $\epsilon n/2$ forever-honest nodes attempts to send $\Terminate$,and each has a $\comsize/n$ probability to be eligible.
The probability that none of them is eligible is 
$(1-\comsize/n)^{\epsilon n/2} < \exp(-\epsilon\comsize/2) = \exp(-\Omega(\comsize))$.
Note that the adversary can fully control in what order honest nodes terminate, 
but it cannot predict which honest nodes are eligible to send $\Terminate$.
Thus, it cannot bias the above probability.
Except for this exponentially small probability,
a $\Terminate$ message sent by an honest eligible node makes all
so-far-honest nodes terminate by the end of the next round. 
\end{proof}

\begin{theorem}[Efficiency]
In expectation, all honest nodes terminate in $O(1)$ rounds and
collectively send
$O(n\comsize)$ messages (i.e., $O(\comsize)$ multicasts). 
\label{thm:efficiency}
\end{theorem}

\begin{proof}
For any iteration, if at least $\epsilon n/2$ forever-honest nodes have terminated,
then by Lemma~\ref{lemma:terminate}, all so-far-honest nodes terminate by the end of the next round except for $\exp(\Omega(-\comsize))$ probability.

Else, by Lemma~\ref{lemma:honest-unique-leader},
with at least $\Theta(1)-\exp(-\Omega(\comsize))=\Theta(1)$ probability, 
a single so-far-honest node (and no already-corrupt node) sends $\Propose$;
further, at least $(\frac12+\frac\epsilon2)n$ so-far-honest nodes have not terminated, and by Chernoff, $\lambda/2$ so-far-honest nodes send $\Status/\Vote/\Commit$ messages;
in this case, all so-far-honest nodes terminate.  
The expected constant round complexity thus follows in a straightforward fashion.
In each round, in expectation, at most $\comsize$ so-far-honest
nodes multicast messages.
Thus, honest nodes send expected $O(n\comsize)$ messages. 
\end{proof}

\begin{corollary}[Efficiency]
Except for $\exp(-\Omega(\numrounds))$ probability, 
all honest nodes terminate in $O(\numrounds)$ rounds
and collectively send $O(n\numrounds^2)$ messages
 (i.e., $O(\numrounds^2)$ multicasts). 
\label{coro:efficiency}
\end{corollary}
\begin{proof}
The probability that none of the $\numrounds$ iterations is good is
$(1-\Theta(1))^\numrounds=\exp(-\Omega(\numrounds))$. 
By Chernoff, except for $\exp(-\Omega(\numrounds))$ probability, 
in each round of each iteration, 
$O(\lambda)$ so-far-honest nodes send messages. 
\end{proof}

\begin{theorem}
For any constant $0 < \epsilon < 1/2$, 
the protocol in this section solves Byzantine agreement 
with $1-\exp(-\Omega(\comsize))$ probability;
the protocol terminates in expected $O(1)$ rounds,
and honest nodes collectively send $O(\lambda)$ messages in expectation.
\end{theorem}

\begin{proof}
Follows from
Theorem~\ref{thm:safety},~\ref{thm:validity},~and~\ref{thm:efficiency}. 
%The multicast complexity accounts for the computational security parameter $\chi$. 
%Thus, the error probability is
%$\negl(\compsec)$ in addition to $\exp(-\Omega(\comsize))$.
\end{proof}

%%% Local Variables:
%%% mode: latex
%%% TeX-master: "podc2019"
%%% End:

\section{Subquadratic BA under Partial Synchrony}
\label{sec:psync13}

In this section, we describe a partially synchronous 
BA protocol that 
tolerates $\frac{1}{3}-\epsilon$ adaptive corruptions.
% and achieves subquadratic communicaton (or sublinear multicast) complexity.

%A more formal definition is included in Section~\ref{sec:defnpsync}.
%\elaine{todo: define this in defn}

\subsection{Warmup: Communication-Inefficient Underlying BA}
\label{sec:psync13:warmup}

Our starting point is a simple quadratic partially synchronous BA protocol in the unknown $\Delta$ model~\cite{dworklynch}.
The protocol is, in fact, similar to the
4-round-per-iteration synchronous protocol in
Section~\ref{sec:sync12:warmup}.
The protocol also runs in iterations and each iteration consists of four steps.
The key change is that every $\numrounds$ iterations, 
all nodes double the step length.
At some point the step length will reach or exceed $\Delta$ rounds.
Until then, messages may or may not arrive within the step.
But after that point, messages sent at the beginning of a
step will be received before the next step,
and the protocol will terminate in expected constant iterations afterwards.

Our description below assumes $n=3f+1$ nodes in total. 
The protocol runs in iterations $r=1,2,\ldots$.
Each iteration runs in four steps, 
called $\Status$, $\Propose$, $\Vote$, and $\Commit$,
respectively.
At the beginning, each step will be of length 1 round.
Every $\numrounds$ iterations, 
all nodes double the length of a step.
All messages are signed.
We denote a collection $2f+1$
(signed) iteration-$r$ $\Vote$ messages 
for the same bit $b \in \{0, 1\}$ from
distinct nodes as an {\it iteration-$r$ certificate}
for $b$ (in comparison, the synchronous protocol in
Section~\ref{sec:sync12:warmup} required only $f+1$ votes). A
certificate from a higher iteration is said to be a 
higher certificate.
For the time being, assume a random leader election oracle that
elects a random leader $L_r$ at the beginning of every iteration $r$.

%Below is the protocol for an iteration $r \geq 2$.
%The protocol for the very first iteration $r=1$ skips the $\Status$ and $\Propose$ rounds. 
\begin{enumerate}[leftmargin=5mm,itemsep=1pt]
\item $\Status$. 
Every node multicasts a $\Status$ message 
of the form $(\Status, r, b, \Cert)$
containing the highest certified bit $b$ it has seen so far
as well as the corresponding certificate $\Cert$.
In the first iteration $r=1$, it will send its signed input bit
to the leader. 
 
\item $\Propose$. The leader $L_r$ chooses a bit $b$ with a 
highest certificate denoted $\Cert$.
The leader multicasts $(\Propose, r, b, \Cert)$. 
If the leader does not have a higher ranked certificate of size $2f+1$, it can
send a certificate $\Cert$ of size $\geq f+1$ each of which is a
signed input bit. We call the latter certificate of size $\geq
f+1$ an input certificate and it is the lowest ranked certificate.
%To unify the presentation, we say that a bit $b$ without any 
%certificate has an iteration-0 certificate and it is treated as the lowest ranked certificate. 

\item $\Vote$. 
%For the first iteration $r=1$, a node votes for its input bit $b$
%by multicasting $(\Vote, r=1, b)$.
%For all iterations $r\geq2$,
If a validly signed $(\Propose, r, b, \Cert)$ message
has been received from $L_r$ with a certificate $\Cert$ for $b$,  
and if the node has not 
observed a higher
certificate for $1-b$, 
it multicasts an iteration-$r$ $\Vote$ message for $b$ 
of the form $(\Vote, r, b)$
with the leader's proposal attached.
%\footnote{The leader's proposal will not be included in a $\Vote$ message or a certificate.}.
%Importantly, if the node has observed a certificate for the opposite bit $1-b$ from the same iteration as \Cert, it \emph{will} vote for $b$.

\item $\Commit$.
If a node has received $2f+1$ iteration-$r$ signed votes 
for the same bit $b$ from distinct nodes
(which form an iteration-$r$ certificate $\Cert$ for $b$), 
it multicasts an iteration-$r$ $\Commit$ message for $b$ of the
form $(\Commit, r, b)$  
with the certificate \Cert attached.
 
\item[$\star$] 
(This step is not part of the iteration and can be executed at any time.)
If a node has received $f+1$ $\Commit$ messages for the same $b$
from the same iteration from distinct nodes,  
it multicasts a termination message of the form $(\Terminate, b)$
with the $f+1$ $\Commit$ messages attached.  
The node then outputs $b$ and terminates. 
This last message will make all other honest nodes multicast the
same $\Terminate$ message, output $b$ and terminate.
\end{enumerate}

%The above protocol is similar to PBFT~\cite{pbft}. The key
%differences are (1) the protocol is executed in a lock-step round
%model, and (2) we assume a
%new leader in every iteration who collects a status of progress
%in the previous iteration, similar to
%Tendermint~\cite{buchman2016tendermint}. 

%The protocol as described does not obtain validity; we discuss
%this separately while addressing validity.
% Our starting point is a simple quadratic partially synchronous BA protocol. 
% The protocol is, in fact, almost identical to the 4-round-per-iteration synchronous protocol in Section~\ref{sec:sync12:warmup}, 
% except for two changes.
% The first change is that every occurrence of $f+1$ (out of $n=2f+1$) matching messages is now replaced with $2f+1$ (out of $n=3f+1$) matching messages of that type.

% The other change involves validity and we discuss it when addressing validity. 
% This new protocol is essentially PBFT~\cite{pbft} with the recent improvement on view change~\cite{hotstuff}.
% We do not repeat the detailed steps.

\paragraph{Consistency.}
The analysis for consistency follows similar arguments to Section~\ref{sec:sync12:warmup}.
We show that 
{\it if any honest node outputs a bit $b$ in iteration~$r$, 
then no certificate for $1-b$ can be formed in iteration $r$ and all subsequent iterations,} 
assuming ideal signatures.

An honest node outputs $b$ in iteration $r \geq 1$,
there must be an iteration-$r$ certificate for $b$.
Recall that a certificate in this protocol consists of $2f+1$
$\Vote$ messages from that iteration.
For an iteration-$r$ certificate for $1-b$ to also exist, 
$f+1$ nodes need to vote for both $b$ and $1-b$.
But there are only $f$ corrupted nodes, a contradiction.
Thus, we can rule out any iteration-$r$ certificate for $1-b$.

Furthermore, $2f+1$ nodes have sent $\Commit$ messages for $b$ in iteration $r$.
This means at least $f+1$ honest nodes have seen the iteration-$r$ certificate for $b$.
The preference for a higher certificate then ensures consistency for all subsequent iterations.
Since no iteration-$r$ certificate for $1-b$ exists, 
those $f+1$ honest nodes will not vote for $1-b$ in iteration $r+1$;
hence, no iteration-$(r+1)$ certificate for $1-b$ can come into existence;
hence, those $f+1$ honest nodes will not vote for $1-b$ in
iteration $r+2$ and so on.
A simple induction completes the proof.

\paragraph{Validity.}
Observe that the protocol starts with signed inputs sent to the leader. Before any higher ranked certificate
can be formed, the leader needs to send a weak certificate of
$f+1$ signed inputs for the same value in its proposal. If all
honest nodes start with the same value $b$, no leader can create
 $\geq f+1$ sized input certificate for value $1-b$,
proving validity.
\ignore{
The protocol as described does not achieve validity. 
We now fix this problem.
Recall that the first iteration skips
$\Status$/$\Propose$ and starts with $\Vote$. A certificate is
formed if $2f+1$ $\Vote$ messages are received in that round.
% each node must wait until it receives $2f+1$ $\Vote$ messages and
% multicasts these $2f+1$ $\Vote$ messages before moving to the
% next round.
\kartik{this weak certificate may not arrive in the unknown GST
  model. Shall we change the protocol so that everyone reports iteration-0
  certs and the leader is allowed to use the 2f+1 iteration-0
  certs (values) as a single cert in any
  subsequent iteration if does not have a higher cert? The leader
  then needs send the majority value among iteration-0
  certs/values/statuses.}
We introduce a special \emph{weak} certificates for iteration-1 only.
A weak certificate only requires $f+1$ $\Vote$ messages (rather than $2f+1$).
Weak iteration-1 certificates are ranked lower than normal iteration-1 certificates;
this way, they do not affect the safety argument above.  
However, they are ranked higher than iteration-0 certificates (i.e., no certificates at all).
Thus, if all honest nodes have the same input bit $b$,
in the first $\Vote$ round, 
% they will all obtain weak iteration-1 certificates for $b$, and 
no iteration-1 certificate (weak or normal) for $1-b$ can exist,
providing validity.
}

\paragraph{Termination.}
\ignore{Due to the above modifications, 
nodes may incur drastically different waiting times in the first $\Vote$ round.
Thus, we also need to re-synchronize nodes after GST.
Such a clock synchronization protocol is given in Abraham et
al.~\cite{abraham2018synchronous}.} 
Once there is an honest leader after length of a step $\geq
\Delta$ rounds,
all honest nodes send to and receive from each other $\Vote$ and
$\Commit$ messages, output and terminate in that iteration. 
Since leaders are selected at random, in expectation, an honest
leader emerges in $O(1)$ iterations once the step size exceeds
$\Delta$. Since the step lengths double every $\numrounds$ iterations,
 the protocol will terminate in expected $O(\numrounds\Delta)$ rounds.

\subsection{Partially Synchronous Subquadratic BA}
\label{sec:psync13:subquad}

Bit-specific eligibility can be added in a fashion similar to Section~\ref{sec:sync12:subquad}.
Node $i$ is eligible to send \Status, \Vote, \Commit, or a
\Terminate message with difficulty $D$, and is eligible to send \Propose with difficulty $D_0$.
All eligibility depends on the bit $b \in \{0, 1\}$ and (with the exception of \Terminate) the iteration number $r$.
$D$ and $D_0$ are appropriate {\it difficulty} parameters 
such that each of the multicast messages except a proposal has a
$\comsize/n$ probability to be eligible and each leader proposal
has a $1/n$ probability to be eligible. As before, we assume $n > \comsize$;
otherwise, one should simply use the quadratic protocol.
 
Now, the subquadratic protocol is almost identical to the warmup protocol 
except for the following changes:
\begin{itemize}[leftmargin=5mm,itemsep=0pt,topsep=1pt]
\item 
every occurrence of \underline{multicast}
is now replaced with ``\underline{conditionally multicast}''; 
\item 
every occurrence of $2f+1$ $\Vote$ or $\Commit$ messages is now replaced with $2\comsize/3$ messages of that type;
\item 
every occurrence of a weak certificate is now
replaced with $\geq \comsize/3$ messages of that type;
\item upon receiving a message of the form $(i, \msg)$ 
(including messages attached with other messages), 
a node invokes $\Fmine.\ver(i, \msg)$ to verify node $i$'s eligibility to send that message.
Note that $\msg$ can be of the form $(\MType, r, b)$
where $\MType \in \{\Status, \Propose, \Vote, \Commit\}$
or of the form $(\Terminate, b)$.
\end{itemize}

\subsection{Proof}
\label{sec:psync13:proof}

The proofs mostly follow the sketch in Section~\ref{sec:psync13:warmup}
and the stochastic process in Section~\ref{sec:sync12:proof}.
As before, we first assume perfect cryptographic primitives.
% and
%account their failure probabilities in the end. 

\begin{theorem}[Consistency]
If an honest node outputs a bit $b$ in iteration~$r$, 
then no certificate for $1-b$ can be formed in iteration $r$ and
all subsequent iterations, except for $\exp(-\Omega(\comsize))$
probability.
\label{thm:safety-psync13}
\end{theorem}

\begin{proof}
An honest node outputs $b$ in iteration $r$, 
only if it has observed $2\comsize/3$ $\Vote$ messages and $2\comsize/3$ $\Commit$ messages for $b$.

If an iteration-$r$ certificate for $1-b$ also exists,
then there are at least $4\comsize/3$ successful mining attemps on $\Vote$ in iteration $r$.
There are at most $(\frac13-\epsilon)n$ eventually-corrupt nodes;
each of them may attempt to vote for both bits.
All remaining nodes are forever-honest and will only attempt to vote for one bit. 
Thus, there are at most $(\frac43-\epsilon)n$ total mining attempts,
each with a probability of $\comsize/n$ to be successful
but at least $4\comsize/3$ attempts succeeded.
By Chernoff, this happens with $\exp(-\Omega(\comsize))$ probability. 

Furthermore,
if some forever-honest node receives $2\lambda/3$ $\Commit$
messages, it implies that $(\frac23 - \epsilon')n$ nodes have
sent $\Commit$ messages except for $\exp(-\Omega(\lambda))$ probability. At least $(\frac13 - \epsilon' + \epsilon)n$ of
these nodes
are forever-honest.
In iteration $r+1$, the remaining $(\frac23 + \epsilon) - (\frac13
- \epsilon' + \epsilon) = (\frac13 + \epsilon')$ may mine for $1-b$ (since if a forever-honest node has sent a $\Commit$ message for $b$ in
iteration $r$, it will not attempt to vote $1-b$ in iteration
$r+1$). In addition, $(\frac13 -
\epsilon)n$ fraction of corrupt nodes mine for $1-b$.
To form a certificate for $1-b$ in iteration $r+1$, $2\comsize/3$ out of the $(\frac23+\epsilon'-\epsilon)n$ mining attempts for $(\Vote, r+1, 1-b)$ from eventually-corrupt nodes need to succeed.
%Once again, only corrupted node will attempt both.
%This again requires at least $4\comsize/3$ out of $(\frac43-\epsilon)n$ total mining attempts to succeed.
Picking $\epsilon' = \epsilon/2$,
by a Chernoff bound, this happens with $\exp(-\Omega(\comsize))$ probability. 
An induction then completes the proof. 
\end{proof}

\begin{theorem}[Validity]
If all honest nodes have the same input bit $b$,
then all nodes will eventually output $b$, 
except for $\exp(-\Omega(\comsize))$ probability.
\label{thm:validity-psync13}
\end{theorem}

\begin{proof}
By a Chernoff bound, except for the said probability,
there will be fewer than $\comsize/3$ signed input bits for $1-b$
(from eventually-corrupt nodes) in any iteration.
So there is no certificate (input certificate or normal certificate) for $1-b$ and hence,
$1-b$ will never be proposed.
Validity then follows from safety.
\end{proof}

We now prove efficiency.
Lemmas~\ref{lemma:honest-unique-leader} from Section~\ref{sec:sync12:proof} still applies.
Lemmas~\ref{lemma:terminate} needs a minor modification as follows but its proof remains almost identical.

\begin{lemma}
Except for $\exp(-\Omega(\comsize))$ probability, 
if at least $\epsilon n/2$ forever-honest nodes have terminated,
all so-far-honest nodes terminate within $\Delta$ rounds.
\label{lemma:terminate-psync13}
\end{lemma}
\begin{proof}
Similar to that of Lemma~\ref{lemma:terminate}.
\end{proof}

\begin{theorem}[Efficiency]
All honest nodes terminate in expected $O(\lambda\Delta)$ rounds
and collectively send $O(n\comsize^2 \log\Delta)$ messages in expectation (i.e., $O(\comsize^2\log\Delta)$ multicasts).
\label{thm:efficiency-psync13}
\end{theorem}

\begin{proof}
Similar to the proof of Theorem~\ref{thm:efficiency},
once the length of a step is at least $\Delta$ rounds,
the protocol terminates in expected $O(1)$ iterations
by Lemma~\ref{lemma:honest-unique-leader}~and~\ref{lemma:terminate-psync13}.
By then, the step length has been doubled at most $\log\Delta$ times.
Hence, there have been at most $\lambda\log\Delta$ iterations,
and at most $\lambda(\Delta + \Delta/2 + \cdots + 2 +
1)=O(\lambda\Delta)$ rounds have passed. 
In each step of each iteration, in expectation,
so-far-honest nodes send $O(\lambda)$ messages. 
\end{proof}

\begin{corollary}[Efficiency]
Except for $\exp(-\Omega(\comsize))$ probability,
all honest nodes terminate in $O(\comsize\Delta)$ rounds
and collectively send $O(n\comsize^2 \log\Delta)$ messages
(i.e., $O(\comsize^2\log\Delta)$ multicasts). 
\end{corollary}
\begin{proof}
Similar to that of Corollary~\ref{coro:efficiency}.
\end{proof}

\begin{theorem}
For any constant $0 < \epsilon < 1/3$, 
the protocol in this section solves Byzantine agreement 
with $1-\exp(-\Omega(\comsize))$ probability;
the protocol terminates in expected $O(\numrounds\Delta)$ rounds,
and honest nodes collectively send $O(\comsize^2 \cdot
\log\Delta)$ messages in expectation.
\end{theorem}

\begin{proof}
Follows from
Theorems~\ref{thm:safety-psync13},~\ref{thm:validity-psync13},~and~\ref{thm:efficiency-psync13}. 
%The multicast complexity accounts for computational security parameter. Similarly, the error probability is
%$\negl(\compsec)$ in addition to $\exp(-\Omega(\comsize))$.
\end{proof}

\paragraph{Remark.} Since we employ cryptography and assume
computationally-bounded adversaries, our network model for
partial synchrony is adopted from
CKPS~\cite{cachin2001secure}. Specifically, we assume that (i) honest
nodes send polynomially many messages (or $\Delta$
is polynomially bounded), and (ii) the delivery of messages is
controlled by an adversary.

\ignore{
Observe that the definition of a partially synchronous
network~\cite{dworklynch} states that (1) the time until GST can
be arbitrarily long, and (2) before GST,
the delivery of messages in the network is \emph{arbitrary}. Regarding
the first statement, we
note that our protocol works except with a negligible failure
probability in a security parameter for every iteration. Hence,
our protocol requires GST to be polynomially bounded. Regarding the
second statement,
our arguments in Lemma~\ref{lemma:terminate-psync13} crucially
assume that the delivery of messages is controlled by an
adversary. This ensures that honest nodes receive the message
with except with exponentially small probability. If message
delivery is indeed \emph{arbitrary}, then it is 
 possible for all terminated nodes to not be able send a
message to other honest nodes.
}

\ignore{
\footnote{We note that the exponentially small
  failure probability holds for every iteration of the
  protocol. Thus, we need to assume that time for GST is
  polynomially bounded.}
\footnote{This argument crucially assumes that the
  delivery of messages is controlled by an adversary. If message
  delivery is \emph{arbitrary}, then it is indeed possible for all terminated nodes to
  not be able send a message.}
}

%%% Local Variables:
%%% mode: latex
%%% TeX-master: "podc2019"
%%% End:

\section{Necessity of Setup Assumptions for Sublinear Multicast Complexity}
\label{sec:lower}

In this section, we show that some form of setup assumption is needed
for multicast-based subquadratic BA.
Specifically, with plain authenticated channels, 
we show the impossibility of sublinear multicast-complexity BA. 
In this model, a message carries the true identity of the sender, i.e.,
the communication channel authenticates the sender,
but no other setup is available.

As mentioned in Section~\ref{sec:model}, proving the lower bound for Byzantine broadcast
makes it stronger (and applicable to BA).
Thus, we restate the lower bound (i.e., Theorem~\ref{thm:intromcclb}) for Byzantine broadcast below.

\begin{theorem}
In a plain authenticated channel model without setup assumptions,
no protocol can solve Byzantine broadcast with $\cc$ multicast complexity
with probability $p > 5/6$ under $\cc$ adaptive corruptions.
\end{theorem}

Although the lower bound is stated for multicast-based protocols,
the same bound applies to a more general class of protocols in which
at most $\cc$ nodes send messages with $p>5/6$ probability.
In addition, the lower bound holds even when assuming the existence of a random oracle or a memory-erasure model.  

\begin{figure}[tb]
  \centering
  \includegraphics[width=0.5\columnwidth]{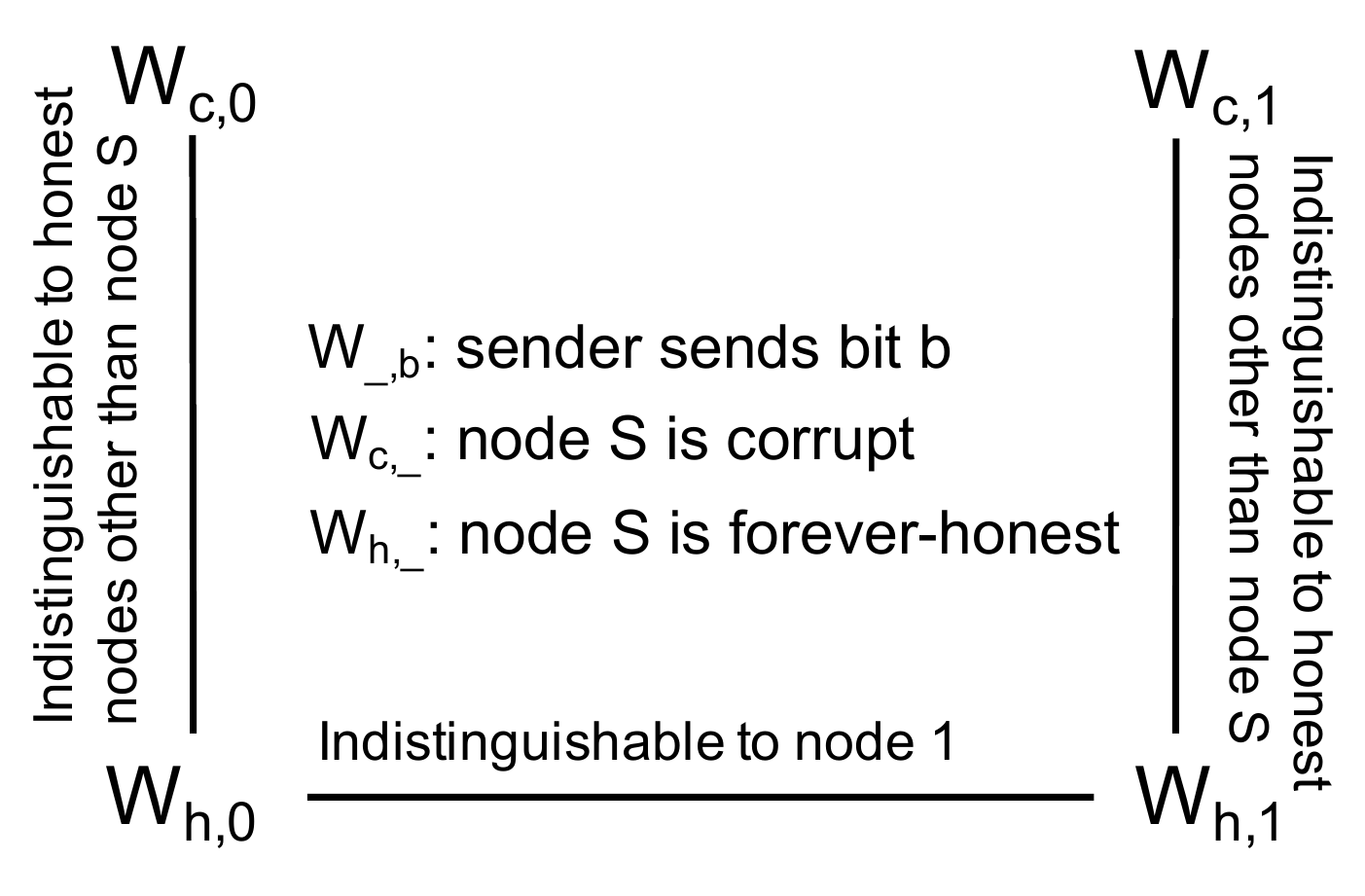}
  \caption{Relationships between different worlds in the
    sublinear multicast complexity without setup assumptions.}
  \label{fig:lowerbound}
\end{figure}

Our proof is inspired
by the classical techniques for proving consensus lower bounds  
in the authenticated channel 
model~\cite{consensuslb,byzgen,weakbyzgen}; however, we extend
known techniques in novel manners, particularly in the way we 
rely on the ability to make adaptive corruptions 
to complete the proof.

\begin{proof}
Suppose for the sake of contradiction that there exists
a protocol that solves Byzantine broadcast using $\cc$ 
multicast complexity with probability $p > 5/6$,
in the authenticated channel model without any trusted setup, 
and tolerating $\cc$ adaptive corruptions.

We focus on a special \nodeS that is not the designated sender.
We consider four worlds:
$\WC{0}$, $\WC{1}$, $\WH{0}$, and $\WH{1}$. 
In world $\WC{*}$, \nodeS is corrupt whereas in world
$\WH{*}$, \nodeS is forever-honest. The designated sender sends bit~$b$ 
in world $W_{*,b}$.

The high-level structure of the proof is depicted in Figure~\ref{fig:lowerbound}.
First, since the designated sender is honest in $\WC{b}$,
with probability $p>5/6$, honest nodes output $b$ in $\WC{b}$ to preserve validity. 
Next, we will show that world $\WC{b}$ and world $\WH{b}$ are indistinguishable to nodes that are forever-honest in both.
Hence, with probability $p>5/6$, these forever-honest nodes output $0$ in $\WH{0}$ and $1$ in $\WH{1}$.
Lastly, we will show that with a constant probability, 
an honest \nodeS cannot distinguish between $\WH{0}$ and $\WH{1}$,
leading to a consistency violation with probability $>1-p$ in one of the two worlds.
Note that the designated sender may be corrupted in $\WH{b}$,
so we need to show a violation of consistency, not validity.

\paragraph{World $\WC{b}$:}
In $\WC{b}$, \nodeS is (statically) corrupt. All other nodes
(including the designated sender) are honest and execute the
 protocol as specified.
The corrupt \nodeS simulates an execution in the $\WC{1-b}$ world
in its head for up to $\cc$ multicasts.
To elaborate, for every round, in addition to receiving messages from honest
nodes in $\WC{b}$, the corrupt \nodeS simulates the
receipt of messages multicast by all other nodes in world $\WC{1-b}$,
until $\cc$ multicasts have occurred in the simulated execution.
The corrupt \nodeS 
treats the received messages (from both the real world $\WC{b}$ and the
simulated world $\WC{1-b}$) as if they are from the same
execution. 
It then sends multicast messages as instructed by an honest
execution of the protocol.
When \nodeS multicasts a message in the real execution, its
messages arrive in both the real as well as the simulated execution.
We note that a simulation of $\WH{1-b}$ by \nodeS is possible
only due to the non-existence of a trusted setup.

Observe that in world $\WC{b}$, only \nodeS is corrupted and
the designated sender is honest.
Hence, by the validity guarantee of the Byzantine broadcast protocol, 
we have:
With probability $p > 5/6$, honest nodes in $\WC{b}$ output $b$.

\paragraph{World $\WH{b}$:}
% We will now describe a world $\WH{b}$ where \nodeS is actually honest but
% behaves like the corrupt \nodeS in $\WC{b}$. 
In $\WH{b}$, all nodes are honest at the start of
the protocol and the adversary makes adaptive corruptions along the way.
The adversary simulates a protocol execution in world $\WH{1-b}$
in its head. Specifically, at the start of each round, the adversary 
simulates this round for all nodes \emph{except \nodeS} in $\WH{1-b}$ in its head, 
and checks to see which nodes will send a message in this
round of the simulated execution. Whenever a node $j$ in the simulated execution
wants to speak, if there have not been $\cc$ multicast messages
from nodes other than \nodeS in this execution,
the adversary adaptively corrupts node $j$
(unless it is already corrupt) in world $\WH{b}$.
If \nodeS multicasts messages in the real
execution, its messages arrive in both the real as well
as the simulated execution.

In a round, a corrupt node $j$ does the following. It sends all
messages as instructed for node $j$ by the protocol. In addition, it
sends the messages node $j$ in $\WH{1-b}$ would have sent to \nodeS in
this round; note that these messages are sent to \nodeS only and not to
anyone else. 

\paragraph{Indistinguishability between worlds $\WH{b}$ and
  $\WC{b}$ for forever-honest nodes.}
The corrupt \nodeS in $\WC{b}$ behaves exactly like the honest \nodeS in $\WH{b}$.
Corrupt nodes in $\WH{b}$ behave honestly towards forever-honest nodes other than \nodeS. 
Therefore, the views of the nodes that are forever-honest in both
 $\WH{b}$ and $\WC{b}$ are identically distributed.
Let $Y$ denote the event that these forever-honest nodes output $b$ in $\WH{b}$.
Based on the indistinguishability and the aforementioned validity guarantee in $\WC{b}$,
we have $\Pr[Y] \geq p > 5/6$.

\paragraph{Indistinguishability between worlds $\WH{b}$ and
  $\WH{1-b}$ for \nodeS.}
Observe that in both worlds, the honest \nodeS receives all the
messages in that world (through the honest protocol execution)
and messages from the first up to $\cc$ nodes in the other world
(through messages sent by 
adaptively corrupted nodes that would be sending honest messages in
the simulated world).
Thus, given that the honest and the simulated execution both
have $\cc$ multicast complexity, the view of \nodeS in worlds 
$\WH{b}$ and $\WH{1-b}$ is identically distributed.

More formally, let $A_r$ and $A_s$ denote the events that the real and simulated
executions respectively have $\cc$ multicast complexity. 
Recall that the protocol satisfies consistency, validity and
termination, and has $\cc$ multicast complexity with
probability $p$. Thus, $\Pr[A_r] \geq p$, $\Pr[A_s] \geq p$, and 
$\Pr[A_r \cap A_s] \geq \Pr[A_r] + \Pr[A_s] - 1 \geq 2p - 1.$

Let $X$ denote the event that \nodeS does not output 1.
Given that the view of \nodeS is identically distributed when the
honest and the simulated executions both have $\cc$
multicast complexity, without
loss of generality, we have $\Pr[X | A_r \cap A_s] \geq 1/2$, and

\begin{align*}
\Pr[X] &\geq \Pr[X \cap A_r \cap A_s] = \Pr[X | A_r \cap A_s] \cdot \Pr[A_r \cap A_s] \\ 
		&\geq \frac{1}{2} (2p-1) > 1/3.
\end{align*}

\paragraph{Consistency violation in $\WH{1}$.}
The probability that consistency of Byzantine
broadcast is violated is given by
\begin{align*}
\Pr[\text{consistency violation}] \geq \Pr[X \cap Y] > 1/3 + 5/6 - 1 = 1/6.  
\end{align*}
This contradicts the supposition that
the protocol solves Byzantine broadcast with $> 5/6$ probability.
\end{proof}

%%% Local Variables:
%%% mode: latex
%%% TeX-master: "podc2019"
%%% End:

%\input{conclusion}

%\paragraph{Acknowledgement.}

{
\bibliographystyle{spmpsci}
\bibliography{refs,crypto}
}

%\appendix
\section{Additional Details on Modeling }
\label{sec:model}

\subsection{Protocol Execution}
\label{sec:execmodel}
In the main body, we 
omitted some details of the execution model for ease
of understanding. We now 
explain these details that will be relevant to the formal proofs.

%We assume a standard protocol execution model  
%with $n$ parties (also called {\it nodes}) numbered $0, 1, \ldots, n-1$.
An external party called the environment and denoted $\algZ$  
provides inputs to honest nodes 
and receives outputs from the honest nodes.
As mentioned, the adversary, denoted $\algA$,
can {\it adaptively} corrupt nodes any time 
during the execution. 
All nodes
that have been corrupt are under the control of $\algA$, i.e.,
the messages they receive are forwarded to $\algA$, and $\algA$ 
controls what messages they will send once they become corrupt. 
The adversary $\algA$ and the environment $\algZ$
are allowed to freely exchange messages any time during the execution.
\ignore{
Henceforth, 
at any time in the protocol, nodes that remain honest so far
are referred to as {\it so-far-honest} nodes; nodes that remain
honest till the end of the protocol are referred to as {\it forever-honest}
nodes; nodes that become corrupt before the end of the protocol are referred
to as {\it eventually-corrupt} nodes.
}
Henceforth, we assume that all parties as well as $\algA$ and $\algZ$
are Interactive Turing Machines, and the execution
is parametrized by a security parameter $\kappa$ that is
common knowledge to all parties as well as $\algA$ and $\algZ$.

\paragraph{Notational conventions.}
Since all parties, including the adversary $\algA$
and the environment $\algZ$
are assumed to be non-uniform probabilitic polynomial-time (\ppt) 
Interactive Turing Machines (ITMs), 
protocol execution is assumed to be probabilistic in nature.
We would like to ensure 
that certain security properties such as consistency
and liveness hold for almost all execution traces, assuming
that both $\algA$ and $\algZ$ are polynomially bounded. 

In our subsequent proofs we sometimes use the notation 
$\view$ to denote a randomly sampled execution. 
\ignore{
Henceforth in the paper, we use the notation
$\view \leftarrow \exec^\Pi(\algA, \algZ, \kappa)$ to denote
a sample of the 
randomized execution of the protocol $\Pi$
with $\algA$ and $\algZ$, and security parameter $\kappa \in \N$.
}
The randomness in the execution comes from honest nodes' randomness, 
\algA, and \algZ, and $\view$ 
is sometimes also referred to as an execution trace or a sample path.
%Henceforth in the paper, 
%an execution trace (i.e., a sample
%path obtained from running the probabilistic execution)
%is often denoted as a ${\sf view}$;
%We assume that there is a publicly known security parameter
%denoted $\kappa$ that is provided as input to all parties
%including $\algA$ and $\algZ$; and 
We would like that the fraction
of sample paths that fail to satisfy relevant security properties 
be negligibly small in the security parameter $\kappa$.

\subsection{Ideal Mining Functionality $\Fsort$}
\label{sec:miningfunc}
Earlier %in Section~\ref{sec:subquad}, 
we used \Fsort to describe our ideal-world protocols.
For preciseness we now spell out the details of \Fsort. %in Figure~\ref{fig:fmine}.
\ignore{
we described how to leverage
cryptographic building blocks such as PRFs and NIZKs to realize
random eligibility election.
For all our protocols, it would be convenient to describe 
them assuming such eligibility election is a blackbox
primitive. We thus introduce an ideal 
functionality called \Fsort which, informally speaking, captures
the cryptographic procedures of random eligibility selection.
One can imagine
that \Fsort is a trusted party such that whenever a node
attempts to mine a ticket for a message type $\msg$, \Fsort
flips a random coin with an appropriate probability 
to decide if this mining attempt is successful.
\Fsort stores 
the results of all previous coin flips, such that 
if a node performs another mining attempt for the same \msg later,
the same result will be used. 

Henceforth in our paper, we will first describe all of our 
protocols
in an ideal world assuming the existence of such a trusted party \Fsort 
(also referred to as \Fsort-hybrid protocols in the cryptography
literature~\cite{uc,Canetti2000}).
Later in Appendix~\ref{sec:real}, we will show that 
using the cryptographic techniques described in Section~\ref{sec:subquad}, 
all of our \Fsort-hybrid protocols  
can be instantiated in a real world where \Fsort does not exist.
}

\begin{figure}[t]
\begin{boxedminipage}{\textwidth}
\begin{center}
$\Fsort(1^\kappa, \mathfrak{p})$
\end{center}
The function $\mathfrak{p} : \{0,1\}^* \rightarrow [0,1]$
maps each message to some success probability.

\begin{itemize}[leftmargin=5mm,itemsep=1pt,topsep=1pt]
\ignore{
\item 
For each $\msg$ and $i$, internally produces the following random variable: 
\begin{itemize}[leftmargin=5mm,itemsep=1pt]
\item 
%if a coin flip denoted $\Coin[\msg, i]$ 
%has not been recorded:
let 
$\Coin[\msg, i] := {\sf Bernoulli}(\mathfrak{p}(\msg))$;
%\item 
%return $\Coin[\msg, i]$.
\end{itemize}
}

\item 
On receive ${\tt mine}(\msg)$ from node $i$ for the first time: 
let $\Coin[\msg, i] := {\sf Bernoulli}(\mathfrak{p}(\msg))$ and 
return $\Coin[\msg, i]$.
\item 
On receive ${\tt verify}(\msg, i)$: %a corrupt node $j$: 
if ${\tt mine}(\msg)$ has been called by node~$i$, 
return $\Coin[\msg, i]$; else return $0$.
\end{itemize}
\end{boxedminipage}
\caption{The mining ideal functionality \Fsort.}
\label{fig:fmine}
\end{figure}

\paragraph{$\Fsort$ ideal functionality.}
As shown in Figure~\ref{fig:fmine}, the \Fsort ideal functionality
has two activation points: 
\begin{itemize}[leftmargin=5mm,itemsep=1pt]
\item 
Whenever a node $i$ calls ${\tt mine}(\msg)$
for the first time, \Fsort flips a random coin with appropriate probability to decide
if node $i$ has successfully mined a ticket for $\msg$.
\item 
If node $i$ has called ${\tt mine}(\msg)$
and the attempt is successful, anyone can then call
${\tt verify}(\msg, i)$
to ascertain that indeed $i$ has mined a ticket for $\msg$.
\end{itemize}

Recall in our scheme, different types of messages
are associated with different probabilities, and we assume that this is
hard-wired in $\Fsort$ with the mapping $\mathfrak{p}$  
that maps each message type to an appropriate probability (see Figure~\ref{fig:fmine}). 
This \Fsort functionality is secret since if an so-far-honest node $i$ has 
not attempted to mine a ticket for $\msg$, then  
no corrupt node can learn whether  
$i$ is in the committee corresponding to $\msg$.

\section{Instantiating \Fsort in the Real World} 
\label{sec:real}

%\elaine{TODO: change to adaptively secure NIZK and commitment}

So far, all our protocols have assumed
the existence of an \Fsort 
ideal functionality.
In this section, we describe how to instantiate
the protocols in the real world 
(where \Fsort does not exist) using cryptography.
Technically 
we do not directly realize the ideal functionality \Fsort 
in the sense of Canetti~\cite{uc}
--- instead, we describe a real-world protocol
that preserves all the security properties of the \Fsort-hybrid protocols. 
%that realizes the ideal-world protocol 
%where ``realizes'' means UC-realize
%by the definition of Canetti~\cite{uc}.

\subsection{Preliminary: Adaptively Secure  
Non-Interactive Zero-Knowledge Proofs}

\hidex{
\elaine{double check the defns, copied from elsewhere}
%\elaine{define simulation extractable nizk here.}
}
We use 
$f(\k) \approx g(\k)$ to mean
that there exists a negligible function $\nu(\k)$ 
such that $|f(\k) - g(\k)| < \nu(\k)$.

A non-interactive proof system 
henceforth denoted $\nizk$
for  an NP language $\lang$
consists of  the following algorithms.

\begin{itemize}[leftmargin=5mm,itemsep=0pt]
\item
$\crs  \leftarrow {\sf Gen}(1^\secparam, \lang)$:
Takes in a security parameter $\secparam$, 
a description of the language $\lang$,
and generates a common reference string $\crs$.
\item
$\pi \leftarrow {\sf P}(\crs, \stmt, w)$:
Takes in $\crs$, a statement $\stmt$, a witness $w$ 
such that $(\stmt, w) \in \lang$, and produces a proof $\pi$.
\item
$ b \leftarrow {\sf V}(\crs, \stmt, \pi)$:
Takes in a $\crs$, a statement $\stmt$, and a proof $\pi$,
and outputs $0$~(reject) or $1$~(accept).

\ignore{
\item
$ (\crs_0, \trap_0) \leftarrow {\sf Gen}_0(1^\k, \lang) $:
Generates a simulated common reference string $\crs_0$ and a trapdoor $\trap_0$.
%, and extract key $\ek$.

\item
$ \pi \leftarrow {\sf P}_0(\crs_0, \trap_0, \stmt) $: Uses trapdoor $\trap_0$ to produce a proof $\pi$ without needing a witness.

\item 
$ (\crs_1, \trap_1) \leftarrow {\sf Gen}_1(1^\k, \lang) $:
Generates a simulated common reference string $\crs_1$ and an
extraction trapdoor $\trap_1$
for knowledge extraction as defined later.

\item 
$w \leftarrow {\sf Explain}(1^\k, \crs_1, \trap_1, \stmt, \pi)$:
Given an extraction CRS $\crs_1$, an extraction trapdoor $\trap_1$
a statement $\stmt$ and a proof $\pi$, outputs a witness $w$. 
}
\end{itemize}

\paragraph{Perfect completeness.}
A non-interactive proof system is said to be perfectly complete,
if an honest prover with a valid witness can always convince an honest verifier.
More formally, for any $(\stmt, w) \in \lang$, we have that
\[
\Pr\left[
%\begin{array}{l}
\crs \leftarrow {\sf Gen}(1^\k, \lang), \
\pi \leftarrow {\sf P}(\crs, \stmt, w): 
{\sf V}(\crs, \stmt, \pi)   = 1
%\end{array}
\right] = 1
\]

\paragraph{Non-erasure computational zero-knowledge.}
Non-erasure zero-knowledge requires that under a simulated
CRS, there is a simulated prover that can produce proofs 
without needing the witness. Further, upon obtaining
a valid witness to a statement a-posteriori, the simulated
prover can explain the simulated NIZK with the correct witness.

We say that a proof system $(\keygen, {\sf P}, {\sf V})$
satisfies non-erasure computational zero-knowledge 
iff there exists a probabilistic polynomial time algorithms
$(\keygen_0, {\sf P}_0, {\sf Explain})$ 
such that 
\[
\Pr\left[\crs \leftarrow \keygen(1^\k), \algA^{{\sf Real}(\crs, \cdot, \cdot)}(\crs) = 1\right]
\approx
\Pr\left[(\crs_0, \trap_0) \leftarrow \keygen_0(1^\k), 
\algA^{{\sf Ideal}(\crs_0, \trap_0, \cdot, \cdot)}(\crs_0) = 1\right],
\]
where ${\sf Real}(\crs, \stmt, w)$
runs the honest prover ${\sf P}(\crs, \stmt, w)$ with randomness
$r$ and obtains the proof $\pi$, it then 
outputs $(\pi, r)$; 
${\sf Ideal}(\crs_0, \trap_0, \stmt, w)$
runs the simulated prover 
$\pi \leftarrow {\sf P}_0(\crs_0, \trap_0, \stmt, \rho)$ 
with randomness $\rho$
and without a witness,
and then runs $r \leftarrow {\sf Explain}(\crs_0, \trap_0, \stmt, w, \rho)$ 
and outputs $(\pi, r)$.

%\elaine{copied and pasted from jen's paper.}

\ignore{
\paragraph{Computational zero-knowledge.}
Informally, an NIZK system is 
computationally zero-knowledge
if the proof does not reveal any information about the witness
to any polynomial-time adversary.
More formally,
an NIZK system is said to have computational zero-knowledge, if
for all non-uniform \ppt adversary $\algA$,
\[
\begin{array}{ll}
 \Pr
\left[
\crs \leftarrow \keygen(1^\k, \L):
\algA^{{\sf P}(\crs, \cdot, \cdot)}(\crs)
=1
\right]
\approx
 \Pr
\left[
(\wtcrs, \trap) \leftarrow \overline{\sf Gen}(1^\k, \L):
\algA^{\overline{\sf P}_1(\wtcrs, \trap, \cdot, \cdot)}(\wtcrs)
=1
\right]
\end{array}
\]
In the above, ${\sf P}_1(\wtcrs, \trap, \stmt, w)$  
verifies that $(\stmt, w) \in \L$, and if so, 
outputs $\overline{\sf P}(\wtcrs, \trap, \stmt)$ which simulates
a proof without knowing a witness. 
Otherwise, if $(\stmt, w) \notin \L$, the experiment aborts.
%Note that simulator must specify the reference string before seeing the theorem
%statements.

\paragraph{Computational soundness.}
We say that an NIZK scheme is computationally 
sound
iff for any non-uniform \ppt adversary 
$\algA$, 
it holds that
\[
 \Pr
\left[
\crs \leftarrow \keygen(1^\k, \L), (\stmt, \pi) \leftarrow \algA(\crs):
{\sf V}(\crs, \stmt, \pi) = 1 
\text{ but }
\stmt \notin \lang
\right]
\approx 0
\]
}

\paragraph{Perfect knowledge extration.}
We say that 
 a proof system $(\keygen, {\sf P}, {\sf V})$
satisfies perfect knowledge extraction, 
if there 
exists probabilistic polynomial-time 
algorithms $(\keygen_1, {\sf Extr})$, 
such that for all (even unbounded) adversary $\algA$, 
\[
\Pr\left[\crs \leftarrow \keygen(1^\k): \algA(\crs) = 1\right]
 = 
\Pr\left[(\crs_1, \trap_1) \leftarrow \keygen_1(1^\k): \algA(\crs_1) = 1\right],
\]
and moreover,  
\[
\Pr\left[
(\crs_1, \trap_1) \leftarrow \keygen_1(1^\k);
(\stmt, \pi) \leftarrow \algA(\crs_1); 
w \leftarrow {\sf Extr}(\crs_1, \trap_1, \stmt, \pi):
\begin{array}{l}
{\sf V}(\crs_1, \stmt, \pi) = 1\\
{\rm but}  \ 
(\stmt, w) \notin \mcal{L}
\end{array}
\right] = 0
\]

\subsection{Adaptively Secure Non-Interactive Commitment Scheme}

An adaptively secure non-interactive commitment scheme 
consists of the following algorithms:
\begin{itemize}[leftmargin=5mm,itemsep=0pt]
\item
$\crs  \leftarrow {\sf Gen}(1^\secparam)$:
Takes in a security parameter $\secparam$, 
and generates a common reference string $\crs$.
\item
$C \leftarrow \comm(\crs, v, \rho)$:
Takes in $\crs$, a value $v$, and a random string $\rho$, 
and outputs a committed value $C$.
\item
$b \leftarrow {\sf ver}(\crs, C, v, \rho)$:
Takes in a $\crs$, 
a commitment $C$, a purported opening $(v, \rho)$, 
%a statement $\stmt$, and a proof $\pi$,
and outputs $0$~(reject) or $1$~(accept).
\end{itemize}

\paragraph{Computationally hiding under selective opening.}
We say that a commitment scheme $(\keygen, \comm, {\sf ver})$
is computationally hiding under selective opening, 
iff there exists a probabilistic polynomial time algorithms
$(\keygen_0, {\comm}_0, {\sf Explain})$ 
such that 
\[
\Pr\left[\crs \leftarrow \keygen(1^\k), 
\algA^{{\sf Real}(\crs, \cdot)}(\crs) = 1\right]
\approx
\Pr\left[(\crs_0, \trap_0) \leftarrow \keygen_0(1^\k), 
\algA^{{\sf Ideal}(\crs_0, \trap_0, \cdot)}(\crs_0) = 1\right] 
\]
where ${\sf Real}(\crs, v)$
runs the honest algorithm $\comm(\crs, v, r)$ with randomness
$r$ and obtains the commitment $C$, it then 
outputs $(C, r)$; 
${\sf Ideal}(\crs_0, \trap_0, v)$
runs the simulated algorithm 
$C \leftarrow {\sf comm}_0(\crs_0, \trap_0, \rho)$ 
with randomness $\rho$ and without $v$, 
and then runs $r \leftarrow {\sf Explain}(\crs_0, \trap_0, v, \rho)$ 
and outputs $(C, r)$.

\paragraph{Perfectly binding.}
A commitment scheme 
is said to be perfectly binding iff for every $\crs$ in the support
of the honest CRS generation
algorithm, 
there does not exist $(v, \rho) \neq (v', \rho')$ 
such that 
$\comm(\crs, v, \rho) = \comm(\crs, v', \rho')$.

\begin{theorem}[Instantiation of our NIZK and commitment schemes~\cite{adaptivenizk}]
Assume standard bilinear group assumptions. 
Then, there exists 
a proof system that satisfies perfect completeness, 
non-erasure 
computational zero-knowledge, and perfect knowledge extraction.
Further, there exist 
a commitment scheme that is perfectly binding and 
computationally hiding under selective opening.
\end{theorem}
\begin{proof}
The existence of such a NIZK scheme was shown by  
Groth et al.~\cite{adaptivenizk}
via a building block that they called
\emph{homomorphic proof commitment scheme}.
This building block can also be used to achieve a commitment scheme 
with the desired properties.
\end{proof}

\subsection{NP Language Used in Our Construction}
In our construction, we will use the following NP language $\lang$.
A pair $(\stmt, w) \in \lang$ iff 
\begin{itemize}[leftmargin=5mm,itemsep=0pt]
\item
parse 
$\stmt := (\rho, c, \crs_{\rm comm}, \msg)$, parse $w := (\sk, s)$;
\item
it must hold that 
$c = {\sf comm}(\crs_{\rm comm}, \sk, s)$, 
and $\prf_{\sk}(\msg) = \rho$.

%where ${\sf comm}$ denotes
%a computationally hiding and perfectly binding commitment scheme
%and $\enc$ denotes a perfectly correct, semantically secure public-key 
%encryption scheme.
%\elaine{to check}
\end{itemize}

\subsection{Compiler from \Fsort-Hybrid Protocols to Real-World Protocols}

Our real-world protocol will remove
the \Fsort oracle by leveraging cryptographic building blocks
including a pseudorandom function family,
a non-interactive zero-knowledge proof system
that satisfies computational zero-knowledge
and computational soundness,
%a perfectly correct and semantically secure public-key encryption scheme,
and a perfectly binding and computationally hiding commitment scheme.
%We will assume that a common reference string (CRS) is chosen after
%all nodes register their public keys with the PKI.

Earlier in Section~\ref{sec:intro}, we have described the intuition
behind our approach. Hence in this section we directly provide
a formal description of how to compile our \Fsort-hybrid protocols
into real-world protocols using cryptography.
This compilation works for our previous 
\Fsort-hybrid protocol described in 
Section~\ref{sec:sync12}.
\ignorepsync{and~\ref{sec:psync13}
respectively.}
The high-level idea is to realize an adaptively secure
VRF from adaptively secure PRFs and NIZKs:

\begin{itemize}[leftmargin=5mm]
\item 
{\bf Trusted PKI setup.} 
Upfront, a trusted party 
runs the CRS generation 
algorithms of the commitment and the NIZK scheme
to obtain $\crs_{\rm comm}$ and $\crs_{\rm nizk}$.
It then chooses a secret PRF key for every node,
where the $i$-th node has key $\sk_i$.
It publishes $(\crs_{\rm comm}, \crs_{\rm nizk})$
as the public parameters, and 
each node $i$'s public key denoted $\pk_i$ is computed as a commitment of $\sk_i$
using a random string $s_i$.
The collection of all users' public keys is published to form the PKI, 
i.e., the mapping from each node~$i$ to its public key~$\pk_i$ is public information.
Further, each node $i$ is given the secret key $(\sk_i, s_i)$.

\ignore{
During the registration phase, every (honest) node
registers a public key with the PKI as follows.
The node $i$ first calls the PRF's ${\sf Gen}(1^\kappa)$
algorithm to generate a secret key $\sk_i$.
It then chooses a random value $s_i$, and computes
$c_i := \comm(\sk_i, s_i)$
and registers the resulting $c_i$ with the PKI.
\item 
{\bf Common reference string.}
{\it After} all nodes register their public keys with the PKI,
call $\nizk.{\sf Gen}$ which outputs a common reference
string denoted $\nizk.{\sf crs}$.
Further,
randomly sample from an appropriate domain a random string
denoted $K$, and sample an honest encryption public key 
denoted $\epk$ for a perfectly correct
and semantically secure public-key encryption scheme. 
The tuple $(\nizk.{\sf crs}, K, \epk)$
is output as the common reference string of the protocol.

\elaine{notation conflict: K}
}

\item 
{\bf Instantiating $\mcal{F}_{\rm mine}.{\tt mine}$}.
Recall that in the ideal-world protocol  
a node $i$ calls $\mcal{F}_{\rm mine}.{\tt mine}(\msg)$
to mine a vote for a message $\msg$.
Now, instead, the node $i$ calls $\rho := {\sf PRF}_{\sk_i}(\msg)$, 
%$\ct := \enc_{\epk}(\sk_i, s')$ for a random $s'$,
and computes
the NIZK proof 
$$\pi := \nizk.{\sf P}((\rho, \pk_i, \crs_{\rm comm}, \msg), (\sk_i, s_i))$$
where $s_i$
the randomness used in committing $\sk_i$ during the trusted setup.
%was used during the registration phase to commit to the secret key.
Intuitively, this zero-knowledge proof proves that the evaluation outcome $\rho$ 
is correct w.r.t. the node's public key (which is a commitment
of its secret key).

The mining attempt for $\msg$ 
is considered successful if
$\rho < D_p$
where $D_p$ is an appropriate difficulty parameter 
such that 
any random string of appropriate length 
is less than $D_p$ with probability $p$ --- recall that the parameter $p$ is
selected in a way that depends on the message $\msg$ being ``mined''.
\item 
{\bf New message format.}
Recall that earlier in our \Fsort-hybrid protocols, 
every message multicast by a so-far-honest node $i$
must of one of the following forms:
\begin{itemize}[leftmargin=5mm,itemsep=1pt]
\item 
Mined messages of the form 
$(\msg, i)$ where node $i$ has successfully called $\Fsort.{\tt mine}(\msg)$;
For example, in the synchronous 
honest majority protocol (Section~\ref{sec:sync12}),
$\msg$ 
can be of the form $({\tt T}, r, b)$ where ${\tt T} \in \{{\tt Propose}, 
{\tt Vote}, {\tt Commit}, {\tt Status}\}$, 
$r$ denotes an epoch number, and $b \in \{0, 1, \bot\}$;
or of the form $({\tt Terminate}, b)$.
%\item
%Signed messages of the form 
%$(\msg, i, \sigma)$ where $\sigma$ is a valid signature
%or some relevant evidence on $\msg$ from $i$;
\elaine{i removed signed messages, our ideal protocol does not have sig?}
\item 
Compound messages, i.e., 
a concatenation of the above types of messages.
\end{itemize}

%Signed messages (that are either stand-alone or contained
%in other messages) need not be transformed in the real world.
For every {\it mined} message $(\msg, i)$ that is either stand-alone 
or contained in a compound message, 
in the real-world protocol, 
we rewrite
$(\msg, i)$ as $(\msg, i, \rho, \pi)$ where the terms 
$\rho$ and $\pi$ are defined in the most natural manner:
\begin{itemize}[leftmargin=5mm]
\item 
If $(\msg, i)$ is part of a message that a 
so-far-honest node $i$ wants to multicast,
then the terms $\rho$ and $\pi$ are those
generated by $i$ 
in place of calling $\Fsort.{\tt mine}(\msg)$ 
in the real world (as explained above);
\item 
Else, if 
$(\msg, i)$ is part of a message that a 
so-far-honest node $j \neq i$ wants to multicast, 
it must be that $j$ has received a valid 
real-world tuple $(\msg, i, \rho, \pi)$
where validity will be defined shortly, 
and thus $\rho$ and $\pi$ are simply the terms 
contained in this tuple.
\end{itemize}

%\elaine{nodes may want to include other node's messages}

\item
{\bf Instantiating $\mcal{F}_{\rm mine}.{\tt verify}$.}
In the ideal world, a node would call $\mcal{F}_{\rm mine}.{\tt verify}$
to check the validity of mined messages upon receiving them (possibly
contained in compound messages).
In the real-world protocol, we perform the following instead:
upon receiving the mined message  
$(\msg, i, \rho, \pi)$ that is possibly contained in compound
messages,  
a node can verify the message's validity 
by checking:
\begin{enumerate}[leftmargin=5mm,itemsep=1pt]
\item 
$\rho < D_p$ where $p$ is an appropriate
difficulty parameter that depends on the type of the mined message; 
and
\item 
$\pi$ is indeed a valid NIZK for the statement 
formed by the tuple $(\rho, \pk_i,  \crs_{\rm comm}, \msg)$.
The tuple is discarded unless both checks pass.
\end{enumerate}
\end{itemize}

\subsection{Main Theorems for Real-World Protocols}
After applying the above compiler to our \Fsort-hybrid 
protocols described in Section~\ref{sec:sync12}.
\ignorepsync{and
\ref{sec:psync13}} we obtain our real-world protocol
\ignorepsync{
for the following settings respectively: 
%1) synchronous network with corrupt majority; 
1) synchronous network with $\frac{1}{2} - \epsilon$ fraction 
corrupt, and 2) partially synchronous network
with $\frac{1}{3} - \epsilon$ fraction corrupt.
}
In this section, we present our main theorem statements
for these three settings.  The proofs for these theorems can be derived
by combining the proofs in 
Section~\ref{sec:sync12} and \ref{sec:psync13}
as well as those in the following section, i.e., Appendix~\ref{sec:realproofs}
where will show that the relevant security properties
are preserved in the real world as long as the cryptographic building
blocks are secure.

In theorem statement below, when we say
that ``{\it assume that the cryptographic building blocks employed 
are secure}'', we formally mean that
1) the pseudorandom function family employed is secure;
2) the non-interactive zero-knowledge proof system
that satisfies non-erasure computational zero-knowledge
and perfect knowledge extraction; 
3) the commitment scheme is 
computationally hiding under
selective opening and 
perfectly binding;
and 4) the signature scheme is secure (if relevant).
%We will assume that a common reference string (CRS) is chosen after
%all nodes register their public keys with the PKI.

\begin{theorem}[Sub-quadratic BA under Synchrony]
Let $\pi_{\rm sync}$
be the protocol obtained by applying the above compiler
to the protocol in Section~\ref{sec:sync12}, 
and assume that the cryptographic building blocks employed
are secure.
Then, for any arbitrarily small positive constant
$\epsilon$, any $n \in \N$,
$\pi_{\rm sync}$ satisfies consistency and validity, 
and tolerates $f<(\frac12-\epsilon)n$ adaptive corruptions,
except for $\negl(\kappa)$ probability.
Further, $\pi_{\rm sync}$ achieves 
expected constant round and 
$\chi \cdot \poly\log(\kappa)$ multicast complexity. 
In the above, $\chi$ is a security parameter
related to the hardness of the cryptographic building blocks;
$\chi=\poly(\kappa)$ under standard cryptographic assumptions and 
$\chi=\poly\log(\kappa)$ if we assume sub-exponential security of the cryptographic primitives employed.
\end{theorem}

\begin{theorem}[Sub-quadratic BA under Partial Synchrony]
Let $\pi_{\rm partialsync}$
be the protocol obtained by applying the above compiler
to the protocol in Section~\ref{sec:psync13}, 
and assume that the cryptographic building blocks employed
are secure.
Then, for any arbitrarily small positive constant
$\epsilon$, any $n \in \N$,
$\pi_{\rm partialsync}$ satisfies consistency and validity,
and tolerates $f<(\frac13-\epsilon)n$ adaptive corruptions,
except for $\negl(\kappa)$ probability.
Further, $\pi_{\rm partialsync}$ achieves 
expected $\Delta \cdot \poly\log(\kappa)$ rounds and 
$\chi \cdot \poly\log(\kappa) \cdot \log\Delta$ multicast complexity. 
In the above, $\chi$ is a security parameter
related to the hardness of the cryptographic building blocks;
$\chi=\poly(\kappa)$ under standard cryptographic assumptions and 
$\chi=\poly\log(\kappa)$ if we assume sub-exponential security of the cryptographic primitives employed.
\end{theorem}
\begin{proof}
Proofs for the above two theorems can be obtained
by combining the \Fsort-hybrid analysis 
in Section~\ref{sec:sync12}~or~\ref{sec:psync13} 
with Appendix~\ref{sec:realproofs}
where we show that the relevant security 
properties are preserved in by the real  
world protocol.
\end{proof}

%%% Local Variables:
%%% mode: latex
%%% TeX-master: "podc2019"
%%% End:

\section{Real World is as Secure as the \Fsort-Hybrid World}
\label{sec:realproofs}

\subsection{Preliminary: PRF's Security Under Selective Opening}
Our proof will directly rely on the security of a PRF  
under selective opening attacks. 
We will prove that any secure PRF family is secure under selective opening 
with a polynomial loss in the security.

\paragraph{Pseudorandomness under selective opening.}
We consider a selective opening adversary that interacts with a challenger.
The adversary can request to create new PRF instances, query
existing instances with  
specified messages, 
selectively corrupt instances and obtain the secret keys of these instances,
and finally, we would like to claim that for instances
that have not been corrupt, the adversary is unable to distinguish
the PRFs' evaluation 
outcomes on any future message from random values from an appropriate domain.
More formally, we consider the following game between a challenger
$\algC$
and an adversary $\algA$.

\vspace{6pt}
\noindent \underline{${\sf Expt}^{\algA}_b(1^\kappa)$}:

\begin{itemize}[leftmargin=5mm]
\item 
$\algA(1^\kappa)$ can adaptively interact with $\algC$ through 
the following queries:
\begin{itemize}[leftmargin=5mm]
\item 
{\it Create instance.} The challenger $\algC$ creates a new PRF instance
by calling the honest ${\sf Gen}(1^\kappa)$. 
Henceforth, the instance will be assigned an index that corresponds
to the number of ``create instance'' queries made so far.
The $i$-th instance's secret key will be denoted $\sk_i$.
\item 
{\it Evaluate.}
The adversary $\algA$ specifies an index $i$ that corresponds
to an instance already created and a message $\msg$, and
the challenger computes 
$r \leftarrow {\sf PRF}_{\sk_i}(\msg)$
and returns $r$ to $\algA$.
\item 
{\it Corrupt.} 
The adversary $\algA$ specifies an index $i$, and the challenger $\algC$
returns $\sk_i$ to $\algA$ (if the $i$-th instance has been created).
%\hubert{How about randomness $s_i$?}
\item 
{\it Challenge.}
The adversary $\algA$ 
specifies an index $i^*$ that must have been created and 
a message $\msg$. If $b = 0$, the challenger returns a completely random
string of appropriate length.
If $b = 1$, the challenger computes 
$r \leftarrow {\sf PRF}_{\sk_{i^*}}(\msg)$ 
and returns $r$ to the adversary.
\end{itemize}
\end{itemize}

We say that $\algA$ is compliant iff with probability $1$, every 
challenge tuple $(i^*, \msg)$ it submits 
satisfies the following: 1) $\algA$ does not make a corruption query on
$i^*$ 
throughout the game; and 2) $\algA$ does not make any evaluation
query on the tuple $(i^*, \msg)$.

\begin{definition}[Selective opening security of a PRF family]
We say that a PRF scheme 
satisfies pseudorandomness 
under selective opening iff 
for any compliant \ppt adversary $\algA$, 
its views in  
${\sf Expt}^{\algA}_0(1^\kappa)$
and 
${\sf Expt}^{\algA}_1(1^\kappa)$
are computationally indistinguishable.
\label{defn:selectiveopen}
\end{definition}

\begin{theorem}
Any secure PRF family 
satisfies pseudorandomness 
under selective opening by Definition~\ref{defn:selectiveopen} (with 
polynomial loss in the security reduction).
\label{thm:selectiveopen}
\end{theorem}

%\section{Extra Warmup: Selective Opening Security for PRFs}
%\label{sec:selectiveopen}

%\begin{theorem}[Restatement of Theorem~\ref{thm:selectiveopen}]
%Any secure PRF family 
%satisfies pseudorandomness 
%under selective opening by Definition~\ref{defn:selectiveopen} 
%(with polynomial loss in the security reduction).
%\end{theorem}

\begin{proof}
\vspace{3pt} \noindent{\bf Single-selective-challenge selective opening security.}
In the single-selective challenge version of the game, the adversary
commits to a challenge identifier $i^*$ upfront during the security game,
such that later, 
challenge queries can only be made for the committed index $i^*$.

First, we can show that any secure PRF family
would satisfy single-selective-challenge selective opening security.
Suppose that there is an efficient adversary $\algA$ that can 
break the single-selective-challenge selective opening security game
for some PRF family.
We construct a reduction $\algR$ that leverages $\algA$ to break
the PRF's security.
The reduction $\algR$ interacts with a PRF challenger as well as $\algA$.
$\algR$ generates PRF keys for all instances other than $i^*$
and answers 
non-$i^*$ 
evaluation and corruption queries honestly.
For $i^*$, $\algA$'s evaluation requests 
are forwarded to the 
PRF challenger. 

We consider the following three hybrids:
\begin{enumerate}[leftmargin=5mm]
\item 
The PRF challenger has a real, randomly sampled PRF function from
the corresponding family, and  
$\algR$ answers $\algA$'s challenge queries on $i^*$ with 
random answers;
\item 
The PRF challenger has a random function, and  
$\algR$ answers $\algA$'s challenge queries on $i^*$ by forwarding 
the PRF challenger's answers (or equivalently by relying with random answers);
and 
%\item 
%The PRF challenger has a real, randomly sampled PRF function from
%the corresponding family, and  
%$\algR$ answers $\algA$'s challenge queries on $i^*$ by forwarding 
%random answers;
\item 
The PRF challenger has a real, randomly sampled PRF function from
the corresponding family, and  
$\algR$ answers $\algA$'s challenge queries on $i^*$ by forwarding 
the PRF challenger's answers.
\end{enumerate}
It is not difficult to see that $\algA$'s view 
in hybrid 1 is identical to its view in the 
single-selective challenge selective opening security game when 
$b = 0$; its view in hybrid 3
is identical to its view in the 
single-selective challenge selective opening security game when 
$b = 1$.
Due to the 
security of the PRF, it is not difficult to see that any adjacent pair
of hybrids are indistinguishable.

\paragraph{Single-challenge selective opening security.}
In the single-challenge selective opening version of the game, 
the adversary can only make challenge queries for a single $i^*$
but it need not commit to $i^*$ upfront at the beginning of the security game. 

We now argue that any PRF that satisfies 
single-selective-challenge selective opening security
must satisfy single-challenge selective opening security
with a polynomial security loss.
The proof of this is straightforward. 
Suppose that there is an efficient adversary $\algA$
that can break the 
single-challenge selective opening security
of some PRF family, 
we can then construct an efficient reduction $\algR$
that breaks the single-selective-challenge selective opening security
of the PRF family.
Basically the reduction $\algR$ guesses at random upfront
which index $i^*$ the adversary  $\algA$ 
will choose for challenge queries. 
$\algR$ then forwards all of $\algA$'s queries
to the challenger of the 
single-selective-challenge selective opening security
security game.
If the guess later turns out to be wrong,
the reduction 
simply aborts and outputs a random guess $b'$.
Otherwise, it outputs the same output as $\algA$.
Suppose that $\algA$ creates $q$ instances of PRFs 
then we can conclude that $\algR$ guesses correctly with probability
at least $1/q$.
Thus whatever advantage $\algA$ has in breaking 
the 
single-challenge selective opening security, 
$\algR$ has 
an advantage that is 
$1/q$ fraction of $\algA$'s advantage in breaking the  
single-selective-challenge selective opening security of the PRF family.

\paragraph{Selective opening security.}
Finally, we show that 
any PRF family that satisfies single-challenge selective opening security
must also satisfy selective opening 
security (i.e., Definition~\ref{defn:selectiveopen})
with a polynomial security loss.
This proof can be completed through a standard hybrid argument in which  
we replace 
the challenge queries from real to random one index at a time (where
replacement is performed for all
queries of the $i$-th new index that 
appeared in some challenge query).
\end{proof}

%The proof of this theorem is standard and deferred to Appendix~\ref{sec:selectiveopen}.

\subsection{Definition of Polynomial-Time Checkable Stochastic Bad Events}
In all of our \Fsort-hybrid protocols earlier, 
%the type of failiures that can lead
some stochastic bad events
related to \Fsort's random coins 
can lead 
to the breach of protocol security (i.e., consistency, validity, or termination)
These stochastic
bad events are of the form imprecisely speaking: 
either there are too few honest mining successes
or there are too many corrupt mining successes.
%if many honest nodes have terminated, no honest node has sent terminate.
%\elaine{is this fine?}
More formally,
for the honest majority protocol, the stochastic
bad events are stated in
Lemmas~\ref{lemma:chernoff},~\ref{lemma:honest-unique-leader},~\ref{lemma:terminate},~and~\ref{lemma:terminate-psync13}.
\ignorepsync{
and for the partially synchronous protocol
the stochastic bad events are stated in
Lemma~\ref{lemma:choice_epochs} and~\ref{lemma:other_msg}.
}
%\elaine{TODO: fix}

%We point out that for the second category, i.e., stochastic bad events,
For these stochastic bad events, 
there is a polynomial-time predicate
henceforth denoted $F$, that takes in 
1) all honest and corrupt mining attempts and the rounds in which the attempts
are made (for a fixed $\view$) and
2) \Fsort's coins as a result of these mining attempts, 
and outputs $0$ or $1$, indicating whether the bad events
are true for this specific $\view$.
Recall that $\view$ denotes an execution trace.

In our earlier \Fsort-world analyses (in Section~%~\ref{sec:dolevstronganalysis},
\ref{sec:sync12}), \ignorepsync{and \ref{sec:psync13}),}
although we have not pointed out this explicitly, but 
our proofs actually suggest that the stochastic bad events
defined by $F$ 
happen with small probability 
{\it even when $\algA$ and \algZ are computationally
unbounded}.

The majority of this section will focus on bounding the second category
of failures, i.e., stochastic bad events  
defined by the polynomial-time predicate $F$ (where $F$ may
be a different predicate for each protocol).

%\paragraph{Additional terminology.}
For simplicity, 
we shall call our \Fsort-hybrid protocol $\protideal$ --- 
for the three different protocols, $\protideal$ is a different protocol;
nonetheless, the same proofs hold for all three protocols. 

\subsection{Hybrid 1}

Hybrid 1 is defined just like our earlier \Fsort-hybrid protocol but
with the following modifications:

\begin{itemize}[leftmargin=5mm]
\item 
\Fsort chooses random PRF keys for all nodes at
the very beginning, and let $\sk_i$ denote
the PRF key chosen for the $i$-th node.
\item 
Whenever a node 
$i$ makes a ${\tt mine}(\msg)$
query, %where $\msg = (\tildeb, \ell, t)$, 
\Fsort determines the outcome of the coin flip as follows:
compute $\rho \leftarrow {\sf PRF}_{\sk_i}(\msg)$
and use $\rho < D_p$ as the coin.

\item 
Whenever $\algA$ adaptively corrupts a node $i$, 
\Fsort  discloses $\sk_i$ to $\algA$.
\end{itemize}

\begin{lemma}
For any \ppt $\AZ$, there exists a negligible function $\negl(\cdot)$ such that
for any $\kappa$, 
the bad events defined by $F$ do not happen
in Hybrid 1 with probability $1 - \negl(\kappa)$.
\label{lem:manyhybrids}
\end{lemma}

\begin{proof}
Let $f$ be the number of adaptive corruptions made by $\algA$.
To prove this lemma we must go through a sequence of inner hybrids
over the number of adaptive corruptions made by the 
adversary $\algA$.

\paragraph{Hybrid $1.f$.}
Hybrid $1.f$ is defined almost identically as 
Hybrid 1 except the following modifications:
Suppose that $\algA$ makes the last corruption query 
in round $t$ and for node $i$.
Whenever the ideal functionality \Fsort 
in Hybrid 1 would have called ${\sf PRF}_{\sk_j}(\msg)$
for any $j$
that is honest-forever and in some round $t' \geq t$, 
in Hybrid $1.f$, 
we replace this call with a random string.

\begin{claim}
Suppose that the PRF scheme satisfies pseudorandomness under selective opening. 
Then, if for any \ppt $\AZ$ 
and any $\kappa$, %the security properties are retained in 
the bad events defined by $F$ do not happen in 
Hybrid $1.f$ with probability 
at least $\mu(\kappa)$, then 
for any 
\ppt $(\algA, \algZ)$ and $\kappa$, 
the bad events defined by $F$ do not happen 
in Hybrid $1$ 
with probability at least $\mu(\kappa) - \negl(\kappa)$.
\end{claim}

\begin{proof}
Suppose for the sake of contradiction that the claim does not hold.
We can then construct a PRF 
adversary $\algA'$
that breaks pseudorandomness under selective opening with non-negligible
probability.
$\algA'$ plays $\Fsort$ when interacting with $\algA$.
$\algA'$ is also interacting with a PRF challenger. 
In the beginning, for every node, $\algA'$
asks the PRF challenger to create a PRF instance for that node.
Whenever
\Fsort needs to evaluate a PRF, 
$\algA'$ forwards the query 
to the PRF challenger. 
This continues until $\algA$ makes the last corruption query, i.e.,
the $f$-th corruption query --- suppose this last corruption
query is made in round $t$ and the node to corrupt is $i$.
At this moment, 
$\algA'$ discloses $\sk_i$ to the adversary. 
However, 
whenever Hybrid 1 
would have needed to compute 
${\sf PRF}_{\sk_j}(\msg)$ 
%\elaine{TODO: FIX, now the msg can be different}
for any $j$ that is honest-forever and in some round $t' \geq t$,
$\algA'$ makes a challenge query to the PRF challenger 
for the $j$-th PRF instance
and on the message queried.
Notice that if the PRF challenger returned random answers to challenges,
$\algA$'s view 
in this interation would be identically distributed as Hybrid $1.f$.
Otherwise, if the PRF challenger returned 
true answers to challenges,
$\algA$'s view 
in this interation would be identically distributed as Hybrid $1$.
\end{proof}

\paragraph{Hybrid $1.f'$.}
Hybrid $1.f'$ is defined almost identically as 
Hybrid $1.f$ except the following modification:
whenever $\algA$ makes the last corruption query ---
suppose that this query is to corrupt node $i$ 
and happens in round $t$ --- 
the ideal functionality \Fsort does not disclose
$\sk_i$ to $\algA$.

\begin{claim}
If for any \ppt $\AZ$ 
and any $\kappa$, %the security properties are retained in 
the bad events defined by $F$ do not happen in 
Hybrid $1.f'$ with probability 
at least $\mu(\kappa)$, then 
for any \ppt $(\algA, \algZ)$ and $\kappa$, %the security properties
%are retained 
the bad events defined by $F$ do not happen in Hybrid $1.f$ 
with probability at least $\mu(\kappa)$.
\end{claim}
\begin{proof}
We observe the following: once the last corruption 
query is made in round $t$ for node $i$,  
given that for any $t' \geq t$, any honest-forever node's coins are 
completely random.
Thus whether or not the adversary receives the last corruption key does not
%help it break the security properties. 
help it to cause the relevant bad events to occur.
Specifically in this case, 
at the moment the last corruption
query is made --- without loss of generality assume
that the adversary makes all possible corrupt mining attempts ---
then 
whether the polynomial-checkable bad events 
defined by $F$ take place
is fully determined 
by \Fsort's random coins 
and independent of any further actions of the adversary at this point.
\end{proof}

\paragraph{Hybrid $1.f''$.}
Hybrid $1.f''$ is defined almost identically as 
Hybrid $1.f'$ except the following modification:
suppose that the last corruption query is to corrupt node $i$ and happens
in round $t$;
whenever the ideal functionality \Fsort in Hybrid $1.f'$ 
would have called ${\sf PRF}(\sk_i, \msg)$ in some round $t' \geq t$ (for the node
$i$ that is last corrupt),
in Hybrid $1.f''$, 
we replace this call's outcome with a random string.

\begin{claim}
Suppose that the PRF scheme satisfies pseudorandomness under selective opening. 
Then, if for any \ppt $\AZ$ 
and any $\kappa$, %the security properties are retained in 
the bad events defined by $F$ do not happen in 
Hybrid $1.f''$ with probability 
at least $\mu(\kappa)$, then 
for any 
\ppt $(\algA, \algZ)$ and $\kappa$, 
%the security properties
%are retained 
the bad events defined by $F$ do not happen
in Hybrid $1.f'$ 
with probability at least $\mu(\kappa) - \negl(\kappa)$.
\label{claim:corrupttorand}
\end{claim}

\begin{proof}
Suppose for the sake of contradiction that the claim does not hold.
We can then construct a PRF 
adversary $\algA'$
that breaks  pseudorandomness under selective opening with non-negligible
probability.
$\algA'$ plays the $\Fsort$ 
when interacting with $\algA$.
$\algA'$ is also interacting with a PRF challenger. 
In the beginning, for every 
node, 
$\algA'$
asks the PRF challenger to create a PRF instance for that node.
Whenever
\Fsort needs to evaluate a PRF, 
$\algA'$ forwards the query 
to the PRF challenger. 
This continues until $\algA$ makes the last corruption query, i.e.,
the $f$-th corruption query --- suppose this last corruption
query is made in round $t$ and the node to corrupt is $i$.
At this moment, 
$\algA'$ does not disclose $\sk_i$ to the adversary and does
not query the PRF challenger to corrupt $i$'s secret key either.
Furthermore, whenever Hybrid $1.f'$ would have called
${\sf PRF}_{\sk_i}(\msg)$ in some round $t' \geq t$, 
%\elaine{FIX: what if the query is from corrupt nodes}
$\algA$ now calls the PRF challenger for the $i$-th PRF instance
and on the specified challenge message, it uses the answer
from the PRF challenger 
to replace 
the 
${\sf PRF}_{\sk_i}(\msg)$ call.
Notice that if the PRF challenger returned random answers to challenges,
$\algA$'s view 
in this interation would be identically distributed as Hybrid $1.f''$.
Otherwise, if the PRF challenger returned 
true answers to challenges,
$\algA$'s view 
in this interation would be identically distributed as Hybrid $1.f'$.
\end{proof}

We can extend the same argument continuing with the 
following sequence of hybrids such that we can replace more and more
PRF evaluations at the end with random coins, and withhold more and more
PRF secret keys from $\algA$ upon adaptive corruption queries --- and nonetheless
the probability that the security properties get broken will not be 
affected too much.

\paragraph{Hybrid $1.(f-1)$.}
Suppose that $\algA$ makes the last but second corruption query for node
$i$ and in round $t$. Now, for any node $j$ that is still honest
in round $t$ (not including node $i$), if ${\sf PRF}_{\sk_j}(\msg)$
is needed by the ideal functionality
in some round $t' \geq t$,
the PRF call's outcome will be replaced with random.
Otherwise Hybrid $1.(f-1)$ is the same as Hybrid $1.f''$.

\begin{claim}
Suppose that the PRF scheme satisfies pseudorandomness under selective opening.
Then, if for any \ppt $\AZ$
and any $\kappa$, %the security properties are retained in
the bad events defined by $F$ do not happen in 
Hybrid $1.(f-1)$ with probability
at least $\mu(\kappa)$, then
for any
\ppt $(\algA, \algZ)$ and $\kappa$, 
%the security properties are retained 
the bad events defined by $F$ do not happen in Hybrid $1.f''$
with probability at least $\mu(\kappa) - \negl(\kappa)$.
\end{claim}
\begin{proof}
Similar to the reduction between the \Fsort-hybrid protocol and Hybrid $1.f$.
\end{proof}

\ignore{
\begin{fact}
If there is a \ppt adversary $\algA$ that
can cause the bad events defined by $F$ to occur 
with probability $\mu$ for Hybrid $1.(f-1)$, 
then there is
another \ppt adversary $\algA'$ such that upon making the last but second
corruption query, it would immediately
make the last corruption query
in the same round as $t$ corrupting an arbitrary node (say, the one
with the smallest index and is not corrupt yet),
and $\algA'$ can
cause the bad events defined by $F$ to occur
with probability at least $\mu$ in Hybrid
$1.(f-1)$.
\label{fct:corruptearly}
\end{fact}
\begin{proof}
Without loss of generality, we can assume that whenever a node
becomes corrupt, it will  
Note that polynomial-time checkable function is insensitive to the naming of nodes. 
Further, for a node 
$i^*$ that is adaptively corrupt after the last corruption query, consider the set
of mining queries 
$i^*$ makes while still being honest after the last corrupt query.
Now, instead of treating this as an honest mining attempt
we treat it as a corrupt mining attempt  
in the input to $F$.
\elaine{FILL}
\end{proof}
}

\paragraph{Hybrid $1.(f-1)'$.}
Almost the same as Hybrid $1.(f-1)$, but without disclosing the secret key
to $\algA$ upon the last but second corruption query.

\begin{claim}
If for any \ppt $\AZ$ 
and any $\kappa$, %the security properties are retained in 
the bad events defined by $F$ do not happen in 
Hybrid $1.(f-1)'$ with probability 
at least $1 - \mu(\kappa)$, then 
for any 
\ppt $(\algA, \algZ)$ and $\kappa$, 
%the security properties are retained 
the bad events defined by $F$ do not happen in Hybrid $1.(f-1)$ 
with probability at least $1 - \mu(\kappa)$.
\end{claim}

\begin{proof}
The proof is similar to the reduction between Hybrid $1.f$ and Hybrid $1.f'$, but
with one more subtlety:
in Hybrid $1.(f-1)$,
upon making the last but second
adaptive corruption query for node $i$ in round $t$,
for any $t' \geq t$ and any node honest in round $t$
(not including $i$ but including the last node to corrupt), all coins are random. 
%By Fact~\ref{fct:corruptearly}, 
%\elaine{explain why below}
Due to this, we observe that 
if there is a \ppt adversary $\algA$ that
can cause the bad events defined by $F$ to occur 
with probability $\mu$ for Hybrid $1.(f-1)$, 
then there is
another \ppt adversary $\algA'$ such that upon making the last but second
corruption query, it would immediately
make the last corruption query
in the same round as $t$ corrupting an arbitrary node (say, the one
with the smallest index and is not corrupt yet),
and $\algA'$ can
cause the bad events defined by $F$ to occur
with probability at least $\mu$ in Hybrid
$1.(f-1)$.

Now, we argue that if such an $\algA'$ can cause the bad events defined
by $F$ to occur in Hybrid $1.(f-1)$
with probability $\mu$, there must be an adversary
$\algA''$  
that can cause the bad events defined by $F$ to occur in Hybrid $1.(f-1)'$
with probability $\mu$ too. 
In particular, $\algA''$ will simply run $\algA'$
until $\algA'$ makes the last but second corruption query. At this point
$\algA''$ 
makes an additional corruption query for an arbitrary 
node that is not yet corrupt. 
At this point, clearly whether
bad events defined by $F$ would occur is independent of any further 
action of the adversary --- and although in Hybrid $1.(f-1)'$, $\algA''$ 
does not get to see the secret key corresponding to the 
last but second query,
it still has the same probability of causing the relevant bad events
to occur as the adversary $\algA'$ in Hybrid $1.(f-1)$.
\end{proof}

\paragraph{Hybrid $1.(f-1)''$.}
Suppose that $\algA$ makes the last but second corruption query for node
$i$ and in round $t$.
Now, for any node $j$ that is still honest
in round $t$ as well as node $j = i$, 
if the ideal functionality needs to call ${\sf PRF}_{\sk_j}(\msg)$
in some round $t' \geq t$
the PRF's outcome will be replaced with random.
Otherwise Hybrid $1.(f-1)''$ is identical to $1.(f-1)'$.

%\begin{claim}
%Suppose that the PRF scheme satisfies pseudorandomness under selective opening. 
Due to the same argument as that of 
Claim~\ref{claim:corrupttorand}, we may conclude that 
if for any \ppt $\AZ$ 
and any $\kappa$, %the security properties are retained in 
the bad events defined by $F$ do not happen in 
Hybrid $1.(f-1)''$ with probability 
at least $\mu(\kappa)$, then 
for any 
\ppt $(\algA, \algZ)$ and $\kappa$, 
the bad events defined by $F$ do not happen
in Hybrid $1.(f-1)'$ 
with probability at least $\mu(\kappa) - \negl(\kappa)$.

In this manner, we define a sequence of hybrids till
in the end, we reach the following hybrid:

\paragraph{Hybrid $1.0$.}
All PRFs evaluations in Hybrid 1 are replaced with random, and
no secret keys are disclosed to $\algA$ upon any adaptive corruption query.

It is not difficult to see that Hybrid $1.0$ is identically distributed
as the \Fsort-hybrid protocol.
We thus conclude the proof of Lemma~\ref{lem:manyhybrids}.

\end{proof}

\subsection{Hybrid 2}
Hybrid 2 is defined almost identically as Hybrid 1, except that now
the following occurs:
\begin{itemize}[leftmargin=5mm,itemsep=1pt]
\item 
Upfront, \Fsort 
generates an honest CRS for the commitment scheme and 
the NIZK scheme
and discloses the CRS to  $\algA$.
%NIZK henceforth denoted $\nizk.\crs$
%and discloses $\nizk.\crs$ to $\algA$.
\item 
Upfront, \Fsort not only chooses secret keys for all nodes, but commits 
to the secret keys of these nodes, 
and reveals the commitments to $\algA$. 
\item 
Every time \Fsort receives a ${\tt mine}$ query 
from a so-far-honest node $i$ 
and for the message $\msg$, 
it evaluates $\rho \leftarrow {\sf PRF}_{\sk_i}(\msg)$
and compute a NIZK proof denoted $\pi$ to vouch for $\rho$.
Now, \Fsort returns $\rho$ and $\pi$ to $\algA$.
\item 
Whenever a node $i$ becomes corrupt, \Fsort
reveals all secret randomness node 
$i$ has used in commitments and NIZKs
so far to $\algA$ in addition to revealing its PRF secret key $\sk_i$.
\end{itemize}

\begin{lemma}
Suppose that the commitment scheme 
is computationally adaptive hiding under selective opening, 
and the NIZK scheme is  
non-erasure computational zero-knowledge,
%adaptively UC-secure,
%\elaine{we need only the adaptive hiding part}
%computationally hiding,
%and that the NIZK scheme satisfies computational zero-knowledge.
Then, 
for any \ppt $\AZ$, there exists a negligible function $\negl(\cdot)$ such that
for any $\kappa$, 
the bad events defined by $F$ do not happen
in Hybrid 2 with probability $1 - \negl(\kappa)$.
\label{lem:zk}
\end{lemma}
\begin{proof}
The proof is standard and proceeds in the following internal hybrid steps.

\begin{itemize}[leftmargin=5mm]
\item 
{\bf Hybrid 2.A.}
Hybrid 2.A is the same as Hybrid 2 but with the following modifications.
\Fsort calls simulated NIZK key generation instead of the real one, 
and for nodes that remain honest so-far, \Fsort simulate their NIZK proofs
without needing the nodes' PRF secret keys.
Whenever an honest 
node $i$ becomes corrupt, \Fsort explains node $i$'s 
simulated NIZKs using node $i$'s 
real $\sk_i$ and randomness used in its commitment, and supplies
the explanations to $\algA$. 

\begin{claim}
%\elaine{FILL: hybrid 2 and hybrid 2.A indistinguishable}
Hybrid 2.A and Hybrid 2 are computationally indistinguishable
from the $\view$ of $\algZ$.
\end{claim}
\begin{proof}
Straightforward due to the non-erasure computational zero-knowledge property
of the NIZK.
\end{proof}

\item 
{\bf Hybrid 2.B.}
Hybrid 2.B is almost identical to 
Hybrid 2.A but with the following modifications.
\Fsort calls the simulated CRS generation for the commitment scheme.
When generating public keys for nodes, it computes simulated
commitments without using the nodes' real $\sk_i$'s.
When a node $i$ becomes corrupt,    
it will use the real $\sk_i$ to compute an explanation for the earlier
simulated commitment.
Now this explanation is supplied to the NIZK's explain algorithm
to explain the NIZK too. 

\begin{claim}
Hybrid 2.A and Hybrid 2.B are computationally indistinguishable 
from the view of the environment $\algZ$.
%\elaine{FILL: hybrid 2.A and hybrid 2.B indistinguishable}
\end{claim}
\begin{proof}
Straightforward by the ``computational hiding under selective opening''
property of the commitment scheme.
\end{proof}
\end{itemize}

\begin{claim}
%\elaine{Fill: Hybrid 2.B is just as secure as Hybrid 1}
If for any \ppt $\AZ$ 
and any $\kappa$, %the security properties are retained in 
the bad events defined by $F$ do not happen in 
Hybrid $1$ with probability 
at least $\mu(\kappa)$, then 
for any \ppt $(\algA, \algZ)$ and $\kappa$, %the security properties
%are retained 
%there exists a negligible function $\negl(\cdot)$,
then the bad events defined by $F$ do not happen in Hybrid $2.B$ 
with probability at least $\mu(\kappa)$. %- \negl(\kappa)$.
%\elaine{is there negl}
\end{claim}
\begin{proof}
Given an adversary $\algA$ 
that attacks Hybrid 2.B, 
we can construct an adversary $\algA'$ that attacks Hybrid 1.
$\algA'$ will run 
$\algA$ internally. 
$\algA'$ runs the simulated CRS generations algorithms for the commitment
and NIZK, and sends the simulated CRSes to $\algA$.
It then runs the simulated commitment scheme and sends
simulated commitments to $\algA$ (of randomly chosen $\sk_i$ for every $i$).
Whenever $\algA$ tries to mine a message,  
$\algA'$ can intercept this mining request,
forward it to its own \Fsort.
If successful, $\algA'$ 
can sample a random number $\rho > D_p$; else
it samples a random number $\rho \leq D_p$. 
It then calls the simulated NIZK prover using $\rho$ 
to simulate a NIZK proof and
sends it to $\algA$.
Whenever $\algA$ wants to corrupt a node $i$,
$\algA'$ corrupts it with its \Fsort, obtains $\sk_i$,
and then runs the ${\sf Explain}$ algorithms of the commitment and NIZK schemes 
and discloses the explanations to $\algA$.
Clearly 
$\algA$'s view in this protocol is identically distributed as in Hybrid 2.B.
Moreover, if $\algA$ succeedings in causing the bad events defined by $F$ to happen,
clearly 
$\algA'$ will too.
\end{proof}

\end{proof}

\subsection{Hybrid 3}
Hybrid 3 is almost identical as Hybrid 2 except 
with the following modifications.
Whenever an already corrupt node
makes a mining query to \Fsort, 
it must supply a $\rho$ and a NIZK proof $\pi$. 
\Fsort then verifies the NIZK proof $\pi$, and if verification
passes, it uses $\rho < D_p$ as the result of the coin flip.

%\hubert{Why is mining query different for corrupt node?}

\begin{lemma}
Assume that the commitment scheme is  
perfectly binding, and the NIZK scheme satisfies perfect knowledge extraction. 
Then, for any \ppt $(\algA, \algZ)$,  
there exists a negligible function $\negl(\cdot)$
such that for any $\kappa$, 
the bad events defined by $F$ do not happen 
in Hybrid 3 except 
with probability $\negl(\kappa)$.
%\elaine{FILL: then bad events don't have in hybrid 3 often}
\end{lemma}
\begin{proof}
We can replace the NIZK's CRS generation $\keygen$
with $\keygen_1$ 
which generates a CRS that is identically 
distributed as the honest $\keygen$, but additionally generates
an extraction trapdoor denoted $\trap_1$. 
%Replace NIZK key gen to generate extraction key for perfect extraction.
Now, upon receiving $\algA$'s NIZK proof $\pi$, \Fsort performs extraction.
The lemma follows by observing that 
due to the perfect knowledge extraction of the NIZK
and the perfect binding property of the commitment scheme, 
it holds except with negligible probability
that the extracted witness does not match the 
node's PRF secret key  
that $\Fsort$ had chosen upfront.
%\elaine{FILL}
\end{proof}

%\elaine{corrupt nodes construct NIZK, use the soundness binding properties.}

In the 
lemma below, when we say
that ``{\it assume that the cryptographic building blocks employed
are secure}'', we formally mean that
the pseudorandom function family employed is secure;
the non-interactive zero-knowledge proof system
that satisfies non-erasure computational zero-knowledge
and perfect knowledge extraction; 
the commitment scheme is 
computationally hiding under
selective opening and 
perfectly binding;
and for the synchronous 
%corrupt majority and 
honest majority protocol, additionally assume that
the signature scheme is secure.

\begin{lemma}
Assume the cryptographic building blocks employed are secure.
Then, for any \ppt $(\algA, \algZ)$,  
there exists a negligible function $\negl(\cdot)$ such that 
for any $\kappa \in \N$, 
relevant security properties (including consistency, validity, and termination) are preseved  
with all but $\negl(\kappa)$ probability in Hybrid 3. 
\end{lemma}
\begin{proof}
As mentioned, only two types of bad events can possibly lead to breach of the relevant
security properties: 1) signature failure; and 2) bad events defined by $F$.
Thus the lemma follows in a straightforward fashion by taking a union
bound over the two.
\end{proof}

\subsection{Real-World Execution}

%We now show that in terms of avoiding bad events defined by $F$, 
We now show that the real-world
protocol is just as secure as Hybrid 3 --- recall
that the security properties we care about include consistency, validity,
and termination.

\begin{lemma}
If there is some \ppt $(\algA, \algZ)$
that causes 
the relevant security properties to be broken 
%the bad events defined by $F$ to occur
in the real world with probability $\mu$, then
there is some \ppt $\algA'$ such that $(\algA', \algZ)$
can cause 
the relevant security properties to be broken
%the bad events 
%defined by $F$ to occur 
in Hybrid 3 with probability at least $\mu$.
\end{lemma}

\begin{proof}
We construct the following $\algA'$:
\begin{itemize}[leftmargin=5mm]
\item 
$\algA'$ obtains %$(\nizk.\crs, K, \epk)$ 
CRSes for the NIZK and the commitment scheme 
from its $\Fsort$ and forwards them to $\algA$.
$\algA'$ also forwards the PKI it learns from \Fsort
to $\algA$.
\item 
Whenever $\algA$ corrupts some node, $\algA'$ does the same
with its \Fsort, and forwards whatever learned to $\algA$.
\item 
Whenever $\algA$ sends some message to an honest node, 
for any portion of the message that is a ``mined message'' of any type,
let $(\msg, \rho, \pi)$ denote this mined message --- we assume that $\msg$ contains
the purported miner of this message denoted $i$.

\begin{itemize}[leftmargin=5mm]
\item 
$\algA'$ checks the validity of $\pi$ and that $\rho < D_p$ for 
an appropriate choice of $p$ depending on the message's type;
ignore the message if the checks fail;
\item 
if the purported sender $i$ is an honest node and 
node $i$ has not successfully mined $\msg$ with \Fsort, 
record a {\sf forgery} event and simply ignore this message.
Otherwise, continue with the following steps.
\item 
if the purported sender $i$ is a corrupt node: 
$\algA'$ issues a corresponding mining
attempt to $\Fsort$ on behalf of $i$ 
with the corresponding $\rho$ and $\pi$ 
if no such mining attempt has been made before;
\item 
Finally, $\algA'$ forwards $\msg$ to the 
destined honest on behalf of the corrupt sender. 
\end{itemize}
\item 
Whenever $\algA'$ 
receives some message from an honest node (of Hybrid 3):
for every portion of the message that 
is a ``mined message'' of any type,  
at this point $\algA'$
must have heard from \Fsort 
the corresponding $\rho$, 
and $\pi$ terms.  
$\algA'$ augments the message with these terms and forwards the resulting
message to $\algA$.
\end{itemize}

Note that conditioned on $\view$s (determined by all randomness
of the execution) with no {\sf forgery} event 
then either the relevant bad events occur
both in Hybrid 3 and the real-world execution, or occur in neither. 
For $\view$s with {\sf forgery} events, 
it is not difficult to see that if Hybrid 3 (on this $\view$)
does not incur the relevant bad events, 
then neither would the real-world 
execution (for this $\view$).
\end{proof}

%%% Local Variables:
%%% mode: latex
%%% TeX-master: "podc2019"
%%% End:

%\input{selectiveopen}

\end{document}